\documentclass[14pt]{article}
\usepackage{amssymb, amsmath}
\usepackage{amsthm, mathrsfs}                                        
\usepackage{verbatim}         
\usepackage{graphicx}                               
\usepackage{cite} 
\usepackage{ifthen}  
\usepackage{tikz}
\usepackage{pgfplots}
\usepackage{comment, color} 
\usepackage{algorithm, algpseudocode}
\usepackage[]{float,latexsym,times}
\usepackage{amsfonts,amstext,amsmath,amssymb,amsthm}
\usepackage{caption2}

\newcommand{\mdots}{,\ldots,}

\makeatletter 
\newcommand*{\inlineequation}[2][]{%
  \begingroup
    \refstepcounter{equation}%
    \ifx\\#1\\%
    \else
      \label{#1}%
    \fi
    \relpenalty=10000 %
    \binoppenalty=10000 %
    \ensuremath{%
      #2%
    }%
    ~\@eqnnum
  \endgroup
}
\makeatother

\newcommand{\dpmse}{{\rm{SE_{\text{d-PM}}}}}
\newcommand{\dpmrmse}{{\rm{RMSE_{\text{d-PM}}}}}
\newcommand{\adpmrmse}{{\rm{ARMSE_{\text{d-PM}}}}}
\newcommand{\despritrmse}{{\rm{RMSE_{\text{d-ESPRIT}}}}}
\newcommand{\adespritrmse}{{\rm{ARMSE_{\text{d-ESPRIT}}}}}
\newcommand{\aespritrmse}{{\rm{ARMSE_{\text{ESPRIT}}}}}
\newcommand{\dPMCent}  					{\scriptsize{Conventional PM}}
\newcommand{\reals}[0]{\mathbb{R}}
\newcommand{\complex}[0]{\mathbb{C}}

\newcommand{\itr}[1]{ \ifthenelse {\equal{#1}{}} { () } { (#1)} }
\newcommand{\cpy}[1]{ \ifthenelse {\equal{#1}{}} { () } { [#1]} }
\newcommand{\chrod}[1]{ \ifthenelse {\equal{#1}{}} { \delta } { \delta_{#1}} }

\newcommand{\expect}[1]{\mathbb{E}[#1]}
\newcommand{\diag}[1]{{\rm{diag}}[#1]}
\newcommand{\dist}[1]{ \tilde{#1} }  
\newcommand{\card}[1]{ {\rm{card}}[#1] }
\newcommand{\eq}{\triangleq}
\newcommand{\mean}[1]{#1}
 
\newcommand{\ones}[1]{ \ifthenelse {\equal{#1}{}} { \pmb{1} } { \pmb{1}_{#1}} }
\newcommand{\zeros}{ \pmb{0} }
\newcommand{\id}{\pmb{I}}
\newcommand{\mRe}{ {\rm{Re}} }
\newcommand{\trace}[1]{{\rm Tr}[#1]}

\newcommand{\msr}[0]{\pmb{z}}
\newcommand{\msrs}{\pmb{Z}}
\newcommand{\msrscal}[0]{{z}}

\newcommand{\dsmsrscal}[0]{\tilde{z}}

\newcommand{\mcov}[0]{\pmb{R}}
\newcommand{\smcov}[0]{\hat{\pmb{R}}}
\newcommand{\dsmcov}[0]{\tilde{\pmb{R}}}

\newcommand{\mcoveige}[0]{{\lambda}}
\newcommand{\mcoveiges}[0]{\pmb{\Lambda}}
\newcommand{\mcoveigv}[0]{\pmb{u}} 
\newcommand{\mcoveigvs}[0]{\pmb{U}}

\newcommand{\smcoveige}[0]{\hat{{\lambda}}}
\newcommand{\smcoveiges}[0]{\hat{\pmb{\Lambda}}}
\newcommand{\smcoveigv}[0]{\hat{\pmb{u}}}
\newcommand{\smcoveigvs}[0]{\hat{\pmb{U}}}

\newcommand{\dsmcoveige}[0]{\dist{{\lambda}}}
\newcommand{\dsmcoveiges}[0]{\dist{\pmb{\Lambda}}}
\newcommand{\dsmcoveigv}[0]{\dist{\pmb{u}}}
\newcommand{\dsmcoveigvs}[0]{\dist{\pmb{U}}}

\newcommand{\sensorsel}[0]{\pmb{T}}

\newcommand{\avgvec}{\pmb{x}}
\newcommand{\avgscal}{x}
\newcommand{\avgmat} [1] { \ifthenelse {\equal{#1}{}} {\pmb{W}} {w_{#1}} }

\newcommand{\neighbor}[1]{ \ifthenelse {\equal{#1}{}} {\mathcal{N}} {\mathcal{N}_{#1}} }

\newcommand{\avgmateige}{\alpha}
\newcommand{\avgmateigv}{ \pmb{\beta} }
\newcommand{\avgmateigvscal}{ {\beta} }

\newcommand{\itrpm}{Q}
\newcommand{\itrac}{P}
\newcommand{\itract}{P_1}
\newcommand{\itractt}{P_2}
\newcommand{\itracttt}{P_3}

\newcommand{\error}{\delta}
\newcommand{\vamone}{\pmb{B}}
\newcommand{\vamtwo}{\pmb{\Gamma}}
\newcommand{\vamfour}{\pmb{\mcoveigvs}}
\newcommand{\vavone}{\pmb{h}}
\newcommand{\rmsevone}{\pmb{\gamma}}
\newcommand{\rmsevtwo}{\pmb{\mu}}
\newcommand{\vamthree}{\tilde{\sensorsel}}

\newcommand{\sepdist}{d}

\newcommand{\doa}{{\theta}}

\newcommand{\dsdoa}{\dist{\doa}}
\newcommand{\doas}{\pmb{\doa}}

\newcommand{\snr}{\rm{SNR}}
\newcommand{\noise}{\pmb{n}}
\newcommand{\noisescal}{{n}}
\newcommand{\steermat}{\pmb{A}}
\newcommand{\steervec}{\pmb{a}}
\newcommand{\sourcesig}{{s}}
\newcommand{\sourcesigs}{\pmb{\sourcesig}}
\newcommand{\noisevariance}{\sigma^2}
\newcommand{\scov}{\pmb{P}}

\newcommand{\refsenloc}{\pmb{\xi}}
\newcommand{\sincosdoa}{\pmb{\kappa}}
\newcommand{\uppersel}{\overline{\pmb{J}}}
\newcommand{\lowersel}{\underline{\pmb{J}}}

\newcommand{\sigsub}[1]{ \ifthenelse {\equal{#1}{}} {\mcoveigvs_{\rm{s}}} {\mcoveigvs_{{\rm{s}}, #1}} }
\newcommand{\noisesub}[1]{ \ifthenelse {\equal{#1}{}} {\mcoveigvs_{\rm{n}}} {\mcoveigvs_{{\rm{n}}, #1}} }
\newcommand{\sigeigs}[1]{ \ifthenelse {\equal{#1}{}} {\mcoveiges_{\rm{s}}} {\mcoveiges_{{\rm{s}}, #1}} }
\newcommand{\noiseeigs}[1]{ \ifthenelse {\equal{#1}{}} {\mcoveiges_{\rm{n}}} {\mcoveiges_{{\rm{n}}, #1}} }
\newcommand{\uppersigsub}[1]{ \ifthenelse {\equal{#1}{}} {\overline{\mcoveigvs}_{\rm{s}}} {\overline{\mcoveigvs}_{{\rm{s}}, #1}} }
\newcommand{\lowersigsub}[1]{ \ifthenelse {\equal{#1}{}} {\underline{\mcoveigvs}_{\rm{s}}} {\underline{\mcoveigvs}_{{\rm{s}}, #1}} }

\newcommand{\espritdelay}{\pmb{\Psi}}
\newcommand{\espritdelayeig}{\psi}

\newcommand{\ssigsub}[1]{ \ifthenelse {\equal{#1}{}} {\hat{\mcoveigvs}_{\rm{s}}} {\hat{\mcoveigvs}_{{\rm{s}}, #1}} }
\newcommand{\suppersigsub}[1]{ \ifthenelse {\equal{#1}{}} {\hat{\overline{\mcoveigvs}}_{\rm{s}}} {\hat{\overline{\mcoveigvs}}_{{\rm{s}}, #1}} }
\newcommand{\slowersigsub}[1]{ \ifthenelse {\equal{#1}{}} {\hat{\underline{\mcoveigvs}}_{\rm{s}}} {\hat{\underline{\mcoveigvs}}_{{\rm{s}}, #1}} }

\newcommand{\sespritdelay}{\hat{\pmb{\Psi}}}
\newcommand{\sespritdelayeig}{\hat{\psi}}

\newcommand{\dssigsub}[1]{ \ifthenelse {\equal{#1}{}} {\dist{\mcoveigvs}_{\rm{s}}} {\dist{\mcoveigvs}_{{\rm{s}}, #1}} }
\newcommand{\dsuppersigsub}[1]{ \ifthenelse {\equal{#1}{}} {\dist{\overline{\mcoveigvs}}_{\rm{s}}} {\dist{\overline{\mcoveigvs}}_{{\rm{s}}, #1}} }
\newcommand{\dslowersigsub}[1]{ \ifthenelse {\equal{#1}{}} {\dist{\underline{\mcoveigvs}}_{\rm{s}}} {\dist{\underline{\mcoveigvs}}_{{\rm{s}}, #1}} }

\newcommand{\dsespritdelay}{\dist{\pmb{\Psi}}}
\newcommand{\dsespritdelayeig}{\dist{\psi}}
\newcommand{\espritdelayeiglv}{\pmb{q}}
\newcommand{\espritdelayeigrv}{\pmb{r}}

\newcommand{\sosone}{\expect{  \error \dsmcoveigv_l({\itrac}) \, \error \dsmcoveigv_m^H({\itrac}) }}
\newcommand{\sostwo}{\expect{  \error \dsmcoveigv_l({\itrac}) \, \error \dsmcoveigv_m^T({\itrac}) }}

\newcommand{\sosonett}{\expect{  \error \dsmcoveigv_i(\itrac) \, \error \dsmcoveigv_j^H(\itrac) }}
\newcommand{\sostwott}{\expect{  \error \dsmcoveigv_i(\itrac) \, \error \dsmcoveigv_j^T(\itrac) }}

\newcommand{\devev}{\dist{u}}

\newcommand{\pdsmcoveigv}[1]{\dsmcoveigv_{#1}({\itrac})}

\newcommand{\pdsmcov}{\dsmcov(\itrac)}

\newcommand{\pvavone}[1]{\vavone_{#1}(\itrac)}
\newcommand{\pvavonet}[2]{\vavone_{#1}^{#2}(\itrac)}
\newcommand{\espritdelayt}{\dist{\pmb{C}}}
\newcommand{\espritdelaytt}{\dist{\pmb{F}}}
\newcommand{\espritdelayst}{\dist{{c}}}
\newcommand{\espritdelaystt}{\dist{{f}}}

\newcommand{\mlw}{1.1}
\newcommand{\mfw}{0.55}

\newcommand{\dESPRITT}					{\scriptsize{$\despritrmse$ $\itrac\!\!=\!\!10$}}
\newcommand{\dESPRITTT}					{\scriptsize{$\despritrmse$ $\itrac\!\!=\!\!20$}}
\newcommand{\dESPRITTTT}				{\scriptsize{$\despritrmse$ $\itrac\!\!=\!\!30$}}

\newcommand{\AdESPRITT}		 			{\scriptsize{$\adespritrmse$ $\itrac\!\!=\!\!10$}}
\newcommand{\AdESPRITTT}				{\scriptsize{$\adespritrmse$ $\itrac\!\!=\!\!20$}}
\newcommand{\AdESPRITTTT}				{\scriptsize{$\adespritrmse$ $\itrac\!\!=\!\!30$}}

\newcommand{\cESPRIT}				    {\scriptsize{Conventional ESPRIT}}
\newcommand{\Rao}   		        	{\scriptsize{$\aespritrmse$}} 
\newcommand{\CRB}				        {\scriptsize{CRB \cite{see2004direction}}}

\newcommand{\RMSE}	 					{\scriptsize{RMSE (degree) }}  
\newcommand{\SNR}						{\scriptsize{SNR (dB) }}
\newcommand{\NumberofSamples}           {\scriptsize{Number of Samples $N$}}

\newcommand{\NormalizedRMSE}            {\scriptsize{RMSE}}

\newcommand{\dPMT}  					{\scriptsize{$\dpmrmse$ $\itrac\!\!=\!\!10$}}
\newcommand{\dPMTT}  					{\scriptsize{$\dpmrmse$ $\itrac\!\!=\!\!20$}}
\newcommand{\dPMTTT}  					{\scriptsize{$\dpmrmse$ $\itrac\!\!=\!\!30$}}

\newcommand{\AdPMT}			 			{\scriptsize{$\adpmrmse$ $\itrac\!\!=\!\!10$}}
\newcommand{\AdPMTT}					{\scriptsize{$\adpmrmse$ $\itrac\!\!=\!\!20$}}
\newcommand{\AdPMTTT}					{\scriptsize{$\adpmrmse$ $\itrac\!\!=\!\!30$}}

\usepackage{titlesec} 

\titleformat{\section}{\large\bfseries}{\thesection.}{.5em}{}
\titlespacing*{\section}{0pt}{*3}{*2}
\titleformat{\subsection}{\normalfont\bfseries}{\thesubsection.}{.5em}{}
\titlespacing*{\subsection} {0pt}{*3}{*2}
\titleformat{\subsubsection}{\normalfont\bfseries}{\thesubsubsection.}{.5em}{}
\titlespacing*{\subsubsection} {0pt}{*3}{*2}

\theoremstyle{plain}
\newtheorem{thm}{Theorem}[] 

\begin{document}
                    
\author{Wassim Suleiman\thanks{This work was supported by the LOEWE Priority Program Cocoon  
(http://www.cocoon.tu-darmstadt.de).
A small part of this work is presented in \cite{suleiman2013decentralized}.}\thanks{%
The project ADEL acknowledges the financial support of the Seventh
Framework Programme for Research of the European Commission under
grant number: 619647.}, Marius Pesavento, Abdelhak M. Zoubir}  

\date{}  
\title{Performance Analysis of the Decentralized Eigendecomposition and ESPRIT Algorithm}         

\maketitle

{\small \noindent\textbf{Abstract:}
 
In this paper, 
we consider performance analysis of the decentralized power method
for the eigendecomposition of the sample covariance matrix
based on the averaging consensus protocol.
An analytical expression of the second order statistics 
of the eigenvectors obtained from 
the decentralized power method 
which is required for computing the mean square error (MSE)
of subspace-based estimators
is presented.
We show that the decentralized power method is not an
asymptotically consistent estimator of the eigenvectors
of the true measurement covariance matrix
unless the 
averaging consensus protocol 
is carried out over an infinitely large number of iterations.
Moreover, we introduce the decentralized ESPRIT algorithm which 
yields fully decentralized
direction-of-arrival (DOA) estimates.
Based on the performance analysis of the decentralized power method,
we derive an analytical expression of the MSE of DOA estimators 
using the decentralized ESPRIT algorithm.
The validity of
our asymptotic results is demonstrated by simulations.}

{\small \noindent\textbf{keywords:}
decentralized eigendecomposition,
power method,
decentralized DOA estimation,
ESPRIT, 
averaging consensus.}


\section{Introduction}

Centralized processing in sensor networks
requires the collection of measurements 
or sufficient statistics 
from all sensor nodes at a fusion center (FC)
before processing to obtain meaningful estimates.
A major drawback of a such centralized processing scheme with a single FC,  
is the existence of 
communication bottlenecks 
in large sensor networks with
multi-hop communications \cite{Scaglione2008, suleiman2013decentralized}.
Averaging consensus (AC) protocols \cite{degroot1974reaching,olfati2004consensus,Xiao2004,olfati2007consensus,xiao2007distributed}
achieve an iterative fully decentralized calculation of the average of scalars
that are distributed over a network of nodes.
AC protocols use only local communications between neighboring nodes,
thus, avoiding multi-hop communication.
Moreover, AC protocols perform computations at the nodes and require no FC. 
Thus, AC protocols eliminate communication bottlenecks.
These attributes of AC protocols make them attractive 
and a fully scalable alternative to centralized processing schemes 
in large sensor networks \cite{Scaglione2008}.

The eigendecomposition of the sample covariance matrix
is required in many applications, such as 
signal detection \cite{wax1985detection, williams1990using, Wax1985},
machine learning \cite{bishop2006pattern},
and DOA estimation 
\cite{Schmidt1986,Roy1989,Stoica1990, pesavento2000unitary, Pesavento2002,see2004direction, parvazi2011new, Suleiman2014searchfree}.
Conventionally, the eigendecomposition is carried out in a centralized fashion,
which hinders its application in large sensor networks.
In \cite{Scaglione2008}, 
an algorithm which achieves 
a fully decentralized eigendecomposition of the sample covariance matrix
is introduced.
This algorithm combines the conventional power method (PM) \cite[p.~450]{golub2012matrix},
which represents a centralized iterative eigendecomposition algorithm,  
with the AC protocol to achieve a fully decentralized eigendecomposition
of the sample covariance matrix.
We refer to this algorithm as the decentralized power method (d-PM).
Analytical expressions of the second order statistics of 
the eigenvectors and eigenvalues of the
conventional (centralized) sample covariance matrix 
are presented in \cite[Theorem~9.2.2]{brillinger2001time}.
The expressions from \cite{brillinger2001time} 
are  asymptotic in the number of samples,
i.e., they become accurate as the number of samples increases.
In \cite{li1993performance}
a different approach, 
which is asymptotic in the effective 
signal-to-noise ratio
$(\snr)$,
is proposed.
This approach holds even for the case of one sample 
if the noise variance is sufficiently small
and can be used to derive a general performance bound for 
DOA estimation \cite{roemer2014analytical, steinwandt2014r}.
However, 
the accuracy of the eigendecomposition 
obtained from the sample covariance matrix using the d-PM
not only suffers from finite sample effects 
and finite PM iterations, 
but also depends on the convergence speed of the AC protocol. 
This additional mismatch is introduced by the decentralized implementation 
and can be mitigated if the  AC protocol is carried out over a large number of iterations.
However, a large number of AC iterations is not always possible
since it is associated with a large communication overhead and latency.
A performance analysis of the decentralized eigendecomposition
which considers estimation errors introduced by the AC protocol is
of wide interest for a large variety of applications.

In \cite{suleiman2013decentralized}, we presented 
a fully decentralized DOA estimation algorithm 
using partly calibrated arrays.
Our DOA estimation algorithm combines 
the d-PM with the conventional ESPRIT algorithm \cite{Roy1989} 
and is thus referred to as the decentralized ESPRIT (d-ESPRIT) algorithm.
The numerical simulations carried out in \cite{suleiman2013decentralized} and
\cite{suleiman2014sam} show that the d-ESPRIT algorithm
achieves similar performance as the conventional ESPRIT algorithm
when a large number of AC iterations is used.
However, an analytical study 
of the performance of the d-ESPRIT algorithm which supports these simulations has not been considered so far.
Moreover, the behavior of the d-ESPRIT algorithm when only a small number of AC
iterations is carried out has not been studied before. 

The first and main contribution of this paper consists in 
the derivation of an analytical expression
of the second order statistics of the eigenvectors for 
the sample covariance matrix computed using the d-PM.
Based on this expression,
we show that the d-PM is not a consistent estimator of the 
eigenvectors of the true measurement covariance matrix,
unless the AC protocol 
is carried out over an infinitely large number of iterations. 
Moreover, we show that when the number of AC iterations used in the d-PM converges to infinity, 
our expression and 
the conventional expression in \cite[Theorem~9.2.2]{brillinger2001time} 
become equivalent.  
The second contribution of this paper consists in 
the derivation of an analytical expression of 
the MSE in DOA estimation using the d-ESPRIT algorithm.

The remainder of this paper is organized as follows.
The measurement model is introduced in Section \ref{sec:model}.
In Section~\ref{sec:d-PM},
we briefly revise the AC protocol \cite{Xiao2004} and the d-PM \cite{Scaglione2008} and 
their main properties, used later in the analysis.
Section~\ref{sec:d-PM-analysis}
considers the performance analysis of the d-PM,
namely the computation of the second order statistics of the eigenvectors resulting from the d-PM.
The d-ESPRIT algorithm and its performance analysis are presented in Section~\ref{sec:d-ESPRIT}.
Simulation results in Section~\ref{sec:simualtion}
compare 
the performance of the d-PM and the d-ESPRIT algorithm with
our analytical expressions of
Sec.~\ref{sec:d-PM-analysis} and
Sec.~\ref{sec:d-ESPRIT}
and with the Cram\'er-Rao Bound (CRB)
for partly calibrated arrays \cite{see2004direction}.

In this paper,
$(\cdot)^T$, 
$(\cdot)^*$
and 
$(\cdot)^H$
denote the 				
transpose, 
complex conjugate and 
the Hermitian transpose operations.
The element-wise (Hadamard) product,
diagonal or block diagonal matrices and
the trace of a matrix are denoted by 
$\odot$, 
$\diag{\cdot}$
and
$\trace{\cdot}$,
respectively. 
The real part and the argument of complex numbers are denoted by
$\mRe[\cdot]$
and 
$\arg[\cdot]$, and $\jmath$ is the imaginary unit.
The expectation operator and the Kronecker delta are expressed as $\expect{\cdot}$
and  
$\chrod{i,j}$,
while
$[\cdot]_{i, j}$, 
$[\cdot]_{i}$,
$\id_i$,
$\zeros_{i}$
and
$\ones{i}$
stands for
the $(i, j)$th entry of a matrix,
the $i$th entry of a vector,
the $i \times i$ identity matrix,
the all zeros vector of size $i$
and
the all ones vector of size $i$,
respectively.
The sets of real and complex numbers are denoted by  
$\reals$ and 
$\complex$.
Variable $x$ at the $q$th PM iteration is expressed as $x^{\itr{q}}$.
The conventional (centralized) estimate of 
$x$ is denoted by $\hat{x}$
whereas its decentralized estimate (computed using the AC protocol) is denoted by $\dist{x}$.
In decentralized estimation,
we distinguish between estimates 
which are available at all nodes,
such an estimate 
at the $k$th node
is denoted by $\dist{x}_{\cpy{k}}$,
and estimates which are distributed among all the nodes, 
where the $k$th node
maintains only a part of the corresponding estimate 
denoted by $\dist{x}_{{k}}$.

\vspace{-2mm}
 
\section{Measurement Model}
\label{sec:model}
We consider a network of $M = \sum_{k=1}^K M_k$ sensors clustered in $K$ nodes,
where the $k$th node comprises of $M_k$ sensors.
The assignment of the sensors to the individual nodes is characterized by 
the sensor selection matrix $\sensorsel$,
whose entries are defined as
\begin{equation}
[\sensorsel]_{i, j} \! \eq \! \left\{\begin{array}{c l}
\!\!1, \!& \text{if the $i$th sensor belongs to the $j$th node}\\
\!\!0, \!& \text{otherwise},
\end{array}\right.
\label{eq:def-t}
\end{equation}
where $i=1 \mdots M$ and $j=1 \mdots K$. 

The measurement vector at the $k$th node at time $t$
is denoted as $\msr_k(t) \in \complex^{M_k \times 1}$
and the overall measurement vector is denoted as 
\vspace{-1mm}
\begin{equation}
\msr(t) \eq [\msr_1^T(t) \mdots \msr_K^T(t)]^T \in \complex^{M \times 1}.
\label{eq:measurements}
\end{equation}
The random measurement vector $\msr(t)$ is assumed to be zero-mean with covariance matrix 
$\mcov \triangleq \expect{ \msr(t) \msr^H(t) }$. 
The eigendecomposition of the true measurement covariance matrix 
$\mcov$
is defined as
\vspace{-1mm}
\begin{equation}
\mcov \eq \mcoveigvs \mcoveiges \mcoveigvs^H,
\label{eq:cov-eig-def} 
\vspace{-1mm}
\end{equation}
where $\mcoveigvs \triangleq [\mcoveigv_1 \mdots \mcoveigv_M]$
and
$\mcoveiges \triangleq \diag{ \mcoveige_1 \mdots \mcoveige_M }$
and
$\mcoveigv_1 \mdots \mcoveigv_M$
are the eigenvectors of $\mcov$, corresponding to the 
eigenvalues $\mcoveige_1 > \cdots > \mcoveige_M$.
For the later use, we define the matrix
\vspace{-1.25mm}
\begin{equation}
\vamone_l \eq \vamfour_{-l} \vamtwo^{-1}_{l} \vamfour_{-l}^H,
\label{eq:def-varmone}
\vspace{-1.25mm}
\end{equation}
where 
$\vamfour_{-l} = [\mcoveigv_1 \mdots \mcoveigv_{l-1}, \mcoveigv_{l+1} \mdots \mcoveigv_{M}]$  
and 
$\vamtwo_{l} = \diag{\mcoveige_1-\mcoveige_l \mdots \mcoveige_{l-1}-\mcoveige_l, \mcoveige_{l+1}-\mcoveige_l \mdots \mcoveige_M-\mcoveige_l}$.

In practice, the covariance matrix $\mcov$ is not available,
and can only be estimated from $N$ observations of $\msr(t), t=1 \mdots N$ as
\begin{equation}
\smcov \eq \frac{1}{N} \sum_{t=1}^{N} \msr(t) \msr^H(t).
\label{eq:def-smcov}
\end{equation}
We refer to the conventional estimator of 
the sample covariance matrix in Eq.~(\ref{eq:def-smcov})
as the centralized estimator, 
since it requires that all measurements from every node are available at a FC. 
Let $\smcoveiges, \smcoveigvs, \smcoveige_i$ and $\smcoveigv_i$
be the estimates of 
$\mcoveiges, \mcoveigvs, \mcoveige_i$ and $\mcoveigv_i$ for $i=1 \mdots M$, respectively,
obtained from the centralized eigendecomposition 
of the sample covariance matrix $\smcov$.

In the following section, the decentralized estimation of
the eigenvalues and eigenvectors of the true covariance matrix $\mcov$
using AC protocol,
i.e., without a FC, is introduced and analyzed.
We denote as
$\dsmcoveiges, \dsmcoveigvs, \dsmcoveige_i$ and $\dsmcoveigv_i$
the decentralized estimates of $\mcoveiges, \mcoveigvs, \mcoveige_i$ and $\mcoveigv_i$
for $i=1 \mdots M$
respectively.

\section{Averaging Consensus and the Decentralized Power Method}
\label{sec:d-PM} 
In this section, 
the AC protocol and its convergence properties are reviewed. 
Moreover, the decentralized eigendecomposition using the d-PM \cite{Scaglione2008} is revised.

\vspace{-1mm}
\subsection{Averaging Consensus}
\label{sec:ac}
Let $\avgscal_1 \mdots \avgscal_K$ denote $K$ scalars
which are available at $K$ distinct nodes in the network,
where the $k$th node stores only the $k$th scalar.
Denote the conventional average of these scalars as $\mean{\avgscal}\eq \frac{1}{K}  \sum_{k=1}^K \avgscal_k$.
In AC protocols \cite{degroot1974reaching,olfati2004consensus,Xiao2004,olfati2007consensus,xiao2007distributed}, 
$\mean{\avgscal}$ is computed iteratively,
where at the $p$th AC iteration,
the $k$th node 
sends its current local estimate of the average $\avgscal_k^{\itr{p-1}}$
to its neighboring nodes, denoted as the set $\neighbor{k}$,
and receives 
the corresponding average estimates of the respective neighboring nodes.
Then, the $k$th node updates its local estimate 
of the average as follows
\vspace{-1mm}
\begin{equation}
\avgscal_k^{\itr{p}} = \avgmat{k, k} \avgscal_k^{\itr{p-1}} + \sum_{i \in \neighbor{k}}
\avgmat{i, k} \avgscal_i^{\itr{p-1}}
\vspace{-2mm}
\label{eq:ac-itr}
\end{equation} 
where $\avgmat{i, k}$ is the weighting factor
associated with the communication link between node $i$ and node $k$,
which satisfies $\avgmat{i, k}=0$ when $i\notin \neighbor{k}$
\cite{Xiao2004}.
The AC iteration in Eq.~(\ref{eq:ac-itr}) is initialized with 
$\avgscal_k^{\itr{0}} = \avgscal_k$ for $k=1 \mdots K$.
For more details,
see \cite{Xiao2004}. 

Denote with $\avgmat{}$ the matrix whose entries are $[\avgmat{}]_{i,j}=\avgmat{i,j}$
for $i,j=1 \mdots K$
and let $\avgvec^{\itr{p}} = [\avgscal_1^{\itr{p}} \mdots \avgscal_K^{\itr{p}}]^T$,
then, the update iteration in Eq.~(\ref{eq:ac-itr}) 
can be expressed as
\vspace{-1mm}
\begin{equation}
\avgvec^{\itr{p}} = \avgmat{} \avgvec^{\itr{p-1}} = \avgmat{}^2 \avgvec^{\itr{p-2}} = 
   \cdots = \avgmat{}^p \avgvec^{\itr{0}}. 
\label{eq:avg-matrix-vector}
\vspace{-1mm} 
\end{equation}
Iteration (\ref{eq:avg-matrix-vector}) converges 
asymptotically (for $p \rightarrow \infty$)
to the vector of averages
$\avgscal \ones{K}$ if
and only if
\vspace{-2mm} 
\begin{equation}
\avgmat{}^{p} \rightarrow \frac{1}{K} \ones{K} \ones{K}^T.
\label{eq:ac-convergence}
\vspace{-1mm}
\end{equation}
Let the eigendecomposition of the matrix $\avgmat{}$ be
\vspace{-1mm}
\begin{equation}
 \avgmat{} \, [\avgmateigv_1 \mdots \avgmateigv_K] =
 [\avgmateigv_1 \mdots \avgmateigv_K] \, \diag{\avgmateige_1 \mdots \avgmateige_K},  
\label{eq:ac-eig}
\vspace{-1mm}
\end{equation}
where $\avgmateige_1 > \cdots > \avgmateige_K$.   
According to \cite{Xiao2004}, 
the matrix $\avgmat{}$ which satisfies the asymptotic convergence condition 
(\ref{eq:ac-convergence}) 
possesses the following properties:
\begin{enumerate}
  \item[P1:] The principle eigenvalue of the matrix $\avgmat{}$ is unique (single multiplicity) 
     and equals to one, i.e., $\avgmateige_1 =1$.
     The corresponding normalized principal eigenvector of the matrix
     $\avgmat{}$ is given by $\avgmateigv_1 = \frac{1}{\sqrt{K}}\ones{K}$.
 \item[P2:] The remaining eigenvalues of $\avgmat{}$ 
    are strictly less than $\avgmateige_1$ in magnitude.
\end{enumerate}
In the following, we assume that the weighting matrix
$\avgmat{}$ satisfies the convergence condition (\ref{eq:ac-convergence}),
which permits the use of properties P1 and P2 in our analysis in Sec.~\ref{sec:d-PM-analysis}.
We express the decentralized estimate
of the average $\mean{\avgscal}$ 
at the $k$th node
using $p$ AC iterations  
as
\vspace{-1mm}
\begin{equation}
\dist{\avgscal}_{\cpy{k}} \eq 
[\avgmat{}^{p} \avgvec^{\itr{0}}]_{k},
\label{eq:notation-ac}   
\vspace{-1mm}
\end{equation}   
where $[\avgmat{}^{p} \avgvec^{\itr{0}}]_{k}$
denotes the $k$th entry of the vector $\avgmat{}^{p} \avgvec^{\itr{0}}$.
The notation $\dist{\avgscal}_{\cpy{k}}$ is used 
since the corresponding average is computed using the AC protocol
and every node stores locally the computed average.

\subsection{The Decentralized Power Method}
\label{sec:dpm}
In this section,
we first review the conventional (centralized) PM \cite[p.~450]{golub2012matrix},
then we review the d-PM \cite{Scaglione2008}.
 
The conventional PM
is an iterative algorithm which can be used to compute the eigendecomposition
of the sample covariance matrix $\smcov$.
Let us assume that $l-1$ eigenvectors of $\smcov$ 
have been computed using the PM,
then the $l$th eigenvector
is computed using the iteration
\vspace{-1mm}
\begin{equation}
	\smcoveigv_l^{\itr{q}}    = (\id_M - \smcoveigvs_{l-1} \smcoveigvs_{l-1}^H) \smcov \, \smcoveigv_l^{\itr{q-1}},
\label{eq:pm-m-th-eigenvector}
\vspace{-1mm}
\end{equation} 
where $\smcoveigv_l^{\itr{q}}$ denotes the $l$th eigenvector of $\smcov$ at the $q$th PM iteration,
$\id\!_M$ is the $M \times M$ identity matrix
and $\smcoveigvs_{l-1} \eq [\smcoveigv_1 \mdots \smcoveigv_{l-1}]$ is the concatenation of the $l-1$
previously computed eigenvectors of $\smcov$.
The vector $\smcoveigv_l^{\itr{0}}$ is a random initial value.
If the PM is carried out for a sufficiently large number of PM iterations $\itrpm$,
then, the normalized vector 
\vspace{-1mm}    
\begin{equation}
\smcoveigv_l = \smcoveigv_l^{\itr{\itrpm}} / \|\smcoveigv_l^{\itr{\itrpm}}\|
\label{eq:pm-normalization}
\vspace{-1mm}
\end{equation}
is the $l$th eigenvector of the matrix $\smcov$.

The d-PM \cite{Scaglione2008}
performs the computations in Eq.~(\ref{eq:pm-m-th-eigenvector})  
and Eq.~(\ref{eq:pm-normalization})
in a fully decentralized fashion based on the AC protocol.
The key idea of the d-PM is to partition the $l$th vector
as
$
\dsmcoveigv_l^{\itr{q}} \eq [\dsmcoveigv_{l,1}^{\itr{q}T} \mdots \dsmcoveigv_{l,K}^{\itr{q}T}]^T
$,   
where the $k$th node stores and updates only
the $k$th part, 
$\dsmcoveigv^{\itr{q}}_{l,k} \in \complex^{M_k \times 1}$,
of the vector $\dsmcoveigv_l^{\itr{q}}$. 
Note that the notation $\dsmcoveigv_{l, k}$ is used,
since the vector $\dsmcoveigv_{l, k}$
is computed using the AC protocol and stored only at the $k$th node.     
In the d-PM, iteration (\ref{eq:pm-m-th-eigenvector}) is split into two steps.
In the first step, the intermediate vector 
\begin{equation}
\begin{aligned}
	\dsmcoveigv_l^{\prime\itr{q}} \eq \smcov \, \dsmcoveigv_l^{\itr{q-1}}
\end{aligned} 
\label{eq:pm-cent-itr}
\end{equation}
is calculated.
In the second step, the vector $\dsmcoveigv_l^{\itr{q}}$ is updated as  
\begin{equation}
	\dsmcoveigv_l^{\itr{q}}    =  \dsmcoveigv_l^{\prime\itr{q}} - \dsmcoveigvs_{l-1}
	\dsmcoveigvs_{l-1}^H \dsmcoveigv_l^{\prime\itr{q}},
\label{eq:dpm-m-th-eigenvector}
\end{equation} 
where $\dsmcoveigvs_{l-1} \eq [\dsmcoveigv_1 \mdots \dsmcoveigv_{l-1}]$.
In the following,
we review how both Steps 
(\ref{eq:pm-cent-itr}) and (\ref{eq:dpm-m-th-eigenvector}) 
and the normalization Step (\ref{eq:pm-normalization}) 
can be carried out in a fully decentralized fashion \cite{Scaglione2008}.

Substituting Eq.~(\ref{eq:def-smcov}) into Eq.~(\ref{eq:pm-cent-itr}),
yields
\begin{equation} 
	\dsmcoveigv_l^{\prime\itr{q}} 
	   =  \frac{1}{N} \sum_{t=1}^{N}  \msr(t) \dsmsrscal^{\itr{q}}_{t, l},
\label{eq:pm-cent-itr-detail} 
\end{equation}  
where $\dsmsrscal^{\itr{q}}_{t, l} \! \eq \!  \msr^H \!(t) \dsmcoveigv_l^{\itr{q-1}}$.
Note that
$
\dsmsrscal^{\itr{q}}_{t,l} \! 
= \! \sum_{k=1}^{K} \msr_k^H \!(t) \dsmcoveigv_{l,k}^{\itr{q-1}}
$,
where $\msr_k^H(t) \dsmcoveigv_{l,k}$ is computed and stored locally at the $k$th node.
Thus, in analogy to (\ref{eq:notation-ac}),
the estimate of $\dsmsrscal^{\itr{q}}_{t, l}$ at the $k$th node 
computed using the AC protocol is  
\begin{equation}
	\dsmsrscal^{\itr{q}}_{t, l, \cpy{k}}  = K 
	\left[  
	   \avgmat{}^{\itrac} \,
	   [\msr_1^H(t)  \dsmcoveigv_{l,1}^{\itr{q-1}} \mdots \msr_K^H(t)  \dsmcoveigv_{l,K}^{\itr{q-1}}]^T 
    \right]_{k},
\label{eq:dpm-mvm-details}
\end{equation}
where $\itrac$ is the number of AC iterations used in this protocol
\cite{Scaglione2008}.
Using $N$ parallel instances of the AC protocol,
the $k$th node 
will locally maintain the scalars 
$\{\dsmsrscal^{\itr{q}}_{t, l, \cpy{k}}\}_{t=1}^N$.
Thus, each node $k$ can locally compute 
one part of the vector $\dsmcoveigv_l^{\prime\itr{q}}$
as 
$\dsmcoveigv_{l,k}^{\prime\itr{q}} = \frac{1}{N} \sum_{t=1}^{N} \msr_k(t)
\dsmsrscal^{\itr{q}}_{t,l, \cpy{k}}$,
that in turn perform  
the first step of the d-PM iteration 
described in Eq.~(\ref{eq:pm-cent-itr}).

Note that in the second step of the d-PM iteration 
only the second term of Eq.~(\ref{eq:dpm-m-th-eigenvector}) 
has to be computed in a decentralized fashion
\cite{Scaglione2008}.
This term can be written as
$\dsmcoveigvs_{l-1} \dsmcoveigvs_{l-1}^H \dsmcoveigv_l^{\prime\itr{q}} = 
	\sum_{i=1}^{l-1} 
	\dsmcoveigv_i^H \devev^{\prime\itr{q}}_{i,l}$,
where $\devev^{\prime\itr{q}}_{i,l} \eq \dsmcoveigv_i^H
\dsmcoveigv_{l}^{\prime\itr{q}}$.
In analogy to (\ref{eq:pm-cent-itr-detail}),
each node can locally compute its corresponding part of 
$\dsmcoveigvs_{l-1} \dsmcoveigvs_{l-1}^H \dsmcoveigv_l^{\prime\itr{q}}$ 
once the scalars $\{ \devev^{\prime\itr{q}}_{i,l} \}_{i=1}^{l-1}$
are available at every node. 
This can be achieved using $l-1$ parallel instances of 
the AC protocol as
\begin{equation}
\devev^{\prime\itr{q}}_{i,l, \cpy{k}} =
 K  
   \left[  
	   \avgmat{}^{\itract} \,
	   [ 
	     \dsmcoveigv_{i, 1}^H \, \dsmcoveigv_{l, 1}^{\prime\itr{q}} 
	     \mdots 
	     \dsmcoveigv_{i, K}^H \, \dsmcoveigv_{l, K}^{\prime\itr{q}}
	   ]^T 
    \right]_{k},
\label{eq:pm-projection-ac} 	 
\end{equation}
where 
$\devev^{\prime\itr{q}}_{i,l, \cpy{k}}$ 
is the $i$th scalar computed at the $k$th node
and
$\itract$ is the number of AC iterations 
used in these $l\!-\!1$ AC protocol instances.
Thus achieving the second step of the d-PM iteration.

After a sufficiently large number of PM iterations $\itrpm$,  
the vector $\dsmcoveigv_l^{\itr{\itrpm}}$  
is normalized as in Eq.~(\ref{eq:pm-normalization}). 
This normalization can be carried out locally once the norm
$\| \dsmcoveigv_l^{\itr{\itrpm}} \|$ is available at each node 
which is achieved using the AC protocol as  
\begin{equation}
\| \dsmcoveigv_l^{\itr{\itrpm}} \|^2_{\cpy{k}} \!=\!
 K \!
	\left[  
	   \avgmat{}^{\itractt} \,
	   [  \dsmcoveigv_{l, 1}^{\itr{\itrpm} H} \, \dsmcoveigv_{l, 1}^{\itr{\itrpm}}
	      \mdots 
	      \dsmcoveigv_{l, K}^{\itr{\itrpm} H} \, \dsmcoveigv_{l, K}^{\itr{\itrpm}}
	   ]^T 
    \right]_{k},
\label{eq:pm-normalization-ac}
\end{equation}
where 
$\itractt$ is the number of iterations used in the AC protocol instance.
Thus, using the d-PM,
the eigendecomposition of the sample covariance matrix can be calculated
without FC.

In the d-PM, communication between the nodes is required 
to compute the scalars in
(\ref{eq:dpm-mvm-details}), 
(\ref{eq:pm-projection-ac}) 
and (\ref{eq:pm-normalization-ac}).
From a signaling perspective, 
the first and most expensive computation 
is that of the $N$ scalars in Eq.~(\ref{eq:dpm-mvm-details}), 
where $N$ AC protocol instances,
i.e., $\itrac N$ AC iterations, are carried out to compute these scalars.
The second most expensive computation lies in Eq.~(\ref{eq:pm-projection-ac})
which requires $l-1$ AC protocol instances.
The third and least expensive computation is 
the normalization of the eigenvectors
which requires only one AC protocol instance.

\section{Performance Analysis of The Decentralized Power Method}
\label{sec:d-PM-analysis}
In this section, we first
reformulate the d-PM as an equivalent centralized PM.
Based on the centralized formulation,
we derive an asymptotic analytical expression 
of the second order statistics of the eigenvectors
for the sample covariance matrix obtained from the d-PM.
Moreover, we show that 
the d-PM is not a consistent estimator 
of the eigenvectors of the true 
measurement covariance matrix $\mcov$.  

\subsection{Assumptions}
Our performance analysis focuses on the errors
resulting from using a finite number of AC iterations $\itrac < \infty$
to compute the scalars  
$\{ \dsmsrscal^{\itr{q}}_{t, l, \cpy{k}} \}_{t=1}^N$
in Eq.~(\ref{eq:dpm-mvm-details}),
because, from a signaling perspective, 
this step represents the most expensive calculation in the d-PM.
Thus, the following assumptions are made:
\begin{itemize}
  \item [A1:] 
  The number of AC iterations $\itract$ and $\itractt$ used 
  to compute the scalars in (\ref{eq:pm-projection-ac})
  and the normalization factors
 in (\ref{eq:pm-normalization-ac}), respectively,  
 are large compared to the number of AC iterations used 
 to compute the scalars $\{ \dsmsrscal^{\itr{q}}_{t, l, \cpy{k}} \}_{t=1}^N$,
 i.e., $\itract \gg  \itrac$ and $\itractt \gg  \itrac$. 
Thus, errors resulting from the finite number of AC iterations 
in Eq.~(\ref{eq:pm-normalization-ac}) and Eq.~(\ref{eq:pm-projection-ac}) 
are negligible compared to 
those in Eq.~(\ref{eq:dpm-mvm-details}).
  \item [A2:] The number of PM iterations $\itrpm$ is sufficiently large such that the errors resulting from the finite number of 
  PM iterations are negligible.
\end{itemize} 

\subsection{Error Expressions for the Decentralized Power Method}

The decentralized eigendecomposition of the sample covariance matrix using
the d-PM yields the vectors 
$\{ \dsmcoveigv_l  \}_{l=1}^M$.
Since under Assumptions A1 and A2 
these vectors depend on $\itrac$
and not on $\itract$, $\itractt$ and $\itrpm$,
we denote them as $\{ \pdsmcoveigv{l}  \}_{l=1}^M$.   
Due to finite AC iteration effects ($\itrac < \infty$), 
these vectors do not exactly correspond to 
the eigenvectors of the matrix $\smcov$.
The following theorem provides further insights 
into the properties of the vectors 
$\{\pdsmcoveigv{l} \}_{l=1}^M$.  
\begin{thm}
Under Assumption A1,
the vectors $\{\pdsmcoveigv{l} \}_{l=1}^M$
are the eigenvectors of the matrix
\begin{equation}
\pdsmcov \triangleq K \left( \sensorsel \avgmat{}^{\itrac}  \sensorsel^T \right) \odot \smcov,
\label{eq:def-rtilde}
\end{equation}
where  
$\sensorsel$ is the sensor selection matrix 
defined in Eq.~(\ref{eq:def-t})
and $\smcov$ is the centralized sample covariance matrix 
defined in Eq.~(\ref{eq:def-smcov}).
\label{thm:pm-covariance}
\end{thm}
\begin{proof}
We prove Theorem~\ref{thm:pm-covariance} by induction.
Thus, first we prove that the vector $\pdsmcoveigv{1}$,
which is computed using the d-PM,
is the principle eigenvector of the matrix $\pdsmcov$.
Then, assuming that the vectors $\{\pdsmcoveigv{i}\}_{i=1}^{l-1}$
are the principle $l-1$ eigenvectors of the matrix $\pdsmcov$,
we prove that 
$\pdsmcoveigv{l}$ is the $l$th eigenvector of the matrix
$\pdsmcov$.
For convenience, we drop the dependency on $P$ from  $\pdsmcov$ and
$\pdsmcoveigv{l}$, throughout the derivations. 

Note that when the d-PM is used to compute the vector $\dsmcoveigv_1$,
then Eq.~(\ref{eq:dpm-m-th-eigenvector}) reduces to
\begin{equation}
\dsmcoveigv_{1}^{\itr{q}} 
=  
\frac{1}{N} \sum_{t=1}^{N} 
   \left[ 
     \dsmsrscal^{\itr{q}}_{t, 1, \cpy{1}} \msr^T_1(t) 
     \mdots
      \dsmsrscal^{\itr{q}}_{t, 1, \cpy{K}} \msr^T_K(t)
   \right]^T,
\label{eq:proof1-1}
\end{equation}
since the matrix $\dsmcoveigvs_{0} = \zeros$.
Let 
\begin{equation*}
\msrs(t) \triangleq \left( \begin{matrix}
\msr_1(t)  		&  \zeros_{M_1}    	& \cdots 		& \zeros_{M_1}			\\ 
\zeros_{M_2}    &  \msr_2(t) 		& \cdots 		& \vdots				\\
\vdots     		&  \vdots    		& \ddots 		& \zeros_{M_{K-1}}		\\ 
\zeros_{M_K} 	&  \zeros_{M_K}		& \cdots 		& \msr_K(t)		
\end{matrix} \right),
\end{equation*}
where $\msr_k(t)$ is defined in Eq.~(\ref{eq:measurements}).
Then, Eq.~(\ref{eq:proof1-1}) is written as
\begin{equation}
\begin{aligned}
	\dsmcoveigv_{1}^{\itr{q}} &= \left( \frac{K}{N}\sum_{t=1}^{N} \msrs(t) 
	 \avgmat{}^{\itrac} \msrs^H(t) \right) \dsmcoveigv_1^{\itr{q-1}} \\
	 & =  \left( \frac{K}{N}\sum_{t=1}^{N} \sum_{k=1}^K 
	\avgmateige_k^{\itrac} \msrs(t) \avgmateigv_k \avgmateigv^H_k \msrs^H(t)
	\right) \dsmcoveigv_1^{\itr{q-1}},
\end{aligned} 
\label{eq:theory1-1}
\end{equation}
where the eigendecomposition of the matrix $\avgmat{}$ in Eq.~(\ref{eq:ac-eig})  is used,
$\sensorsel$ is the sensor selection matrix defined in Eq.~(\ref{eq:def-t}) and 
$\itrac$ is the number of AC iterations used to compute the scalars $\{\dsmsrscal^{\itr{q}}_{t, 1, \cpy{k}} \}_{t=1}^N$.
The product $\msrs(t) \avgmateigv_k$ can be rewritten as
\begin{equation}
\begin{aligned}
\msrs(t) \avgmateigv_k &= \left[
  \avgmateigvscal_{k, 1} \msr^T_{1}(t) \mdots \avgmateigvscal_{k, K} \msr^T_{K}(t) 
\right]^T \\
  &= \left( \sensorsel \avgmateigv_k \right) \odot \msr(t)
\end{aligned}
\label{eq:theory1-1-product}
\end{equation}
where $\avgmateigv_k = \left[ \avgmateigvscal_{k, 1} \mdots \avgmateigvscal_{k, K}\right]^T$.
Substituting Eq.~(\ref{eq:theory1-1-product}) into Eq.~(\ref{eq:theory1-1}), yields
\begin{equation*}
\begin{aligned} 
	\dsmcoveigv_1^{\itr{q}} \!\!\!&=\!\!
	\left(\! \frac{K}{N}\!\sum_{t=1}^{N} \sum_{k=1}^K \! \avgmateige_k^{\itrac}\!
	\left( \left( \sensorsel \avgmateigv_k \right) \! \odot \! \msr(t) \right)
	\left( \left( \sensorsel \avgmateigv_k \right) \! \odot \! \msr(t) \right)^H
	\right) \dsmcoveigv_1^{\itr{q-1}} \\   
	&=\!\! \left( \frac{K}{N}\sum_{t=1}^{N} \sum_{k=1}^K \avgmateige_k^{\itrac}
	 \left( \sensorsel
	 \avgmateigv_k \avgmateigv^H_k  \sensorsel^T \right)
	\!\odot\!
	\left( \msr(t) \msr^H(t) \right)
	 \right)
	 \dsmcoveigv_1^{\itr{q-1}} \\
	 &=\!\! \left( \! K \! \sum_{k=1}^K \! \avgmateige_k^{\itrac} \!
	 \left( \sensorsel
	 \avgmateigv_k \avgmateigv^H_k  \sensorsel^T \right)
	\!\odot \!
	\frac{1}{N}
	\sum_{t=1}^{N}
	\left( \msr(t) \msr^H(t) \right)
	 \!\! \right)
	 \dsmcoveigv_1^{\itr{q-1}} \\
	&=\!\! \left( K \left( \sensorsel \avgmat{}^{\itrac}  \sensorsel^T \right) \odot
	 \smcov \right)  \dsmcoveigv_1^{\itr{q-1}}.
\end{aligned}
\end{equation*}
Thus, the decentralized computation of $\dsmcoveigv_1$ using the d-PM 
can be written as the following iteration    
\begin{equation}
\dsmcoveigv_1^{\itr{q}} = \dsmcov \; \dsmcoveigv_1^{\itr{q-1}},
\label{eq:theorem-dpm-cent-itr-1}
\end{equation}
where $\dsmcov$ is defined in Eq.~(\ref{eq:def-rtilde}).
Note that Eq.~(\ref{eq:theorem-dpm-cent-itr-1})
corresponds to the update procedure of the conventional PM 
applied to the matrix $\dsmcov$.
Thus, after a sufficiently large number of PM iterations $\itrpm$,
the resulting vector $\dsmcoveigv_1^{\itr{\itrpm}}$
converges (if normalized)
to the principle eigenvector of the matrix $\dsmcov$.
It follows from Assumption A1 that the decentralized normalization of 
$\dsmcoveigv_1^{\itr{\itrpm}}$
is accurate.
Thus, under Assumption A1, 
the vector resulting from applying the d-PM to the sample covariance matrix $\smcov$
is the principle eigenvector of the matrix $\dsmcov$ computed using the conventional PM.
This concludes the first part of the induction.

For the second part of the induction,
we assume that the vectors $\{\dsmcoveigv_i\}_{i=1}^{l-1}$
are computed using the d-PM and they are the first $l-1$ eigenvectors of the matrix $\dsmcov$.
Then, we prove the induction for the vector $\dsmcoveigv_l$.

The computation of the vector $\dsmcoveigv_l$ using the d-PM
is achieved as follows.
First, the vector $\dsmcoveigv_l^{\prime\itr{q}}$,
which is defined in Eq.~(\ref{eq:pm-cent-itr})
is computed in a decentralized fashion.
In analogy to Eq.~(\ref{eq:proof1-1}),
$\dsmcoveigv_l^{\prime\itr{q}}$ can be rewritten as
\begin{equation}
\dsmcoveigv_l^{\prime\itr{q}}=\dsmcov \, \dsmcoveigv_l^{\itr{q-1}}.
\label{eq:thm1-eq}
\end{equation}
Second the scalars
$\{\devev^{\prime\itr{q}}_{i,l}\}_{i=1}^{l-1}$
are computed in a decentralized fashion.
Since under Assumption A1 the AC errors
resulting from this computation are negligible,
the decentralized iteration used to compute the vector $\dsmcoveigv_l$ is reduced to
\begin{equation}
\begin{aligned}
	\dsmcoveigv_l^{\itr{q}}  &  = (\id_M - \dsmcoveigvs_{l-1} \dsmcoveigvs_{l-1}^H) \dsmcov \, \dsmcoveigv_l^{\itr{q-1}}.
\end{aligned}
\label{eq:theorem-dpm-mth-eig}
\end{equation}   
Note that Eq.~(\ref{eq:theorem-dpm-mth-eig}) is equivalent to
the conventional PM iteration (\ref{eq:pm-m-th-eigenvector}) applied to compute the $l$th eigenvector of the matrix $\dsmcov$.
After $\itrpm$ iterations of the d-PM, the resulting vector $\dsmcoveigv_l^{\itr{\itrpm}}$
is normalized. 
Again under Assumption A1, the normalization 
is accurate,
thus, the resulting normalized vector $\dsmcoveigv_l$ is the $l$th eigenvector of the matrix $\dsmcov$
computed using the conventional PM.
\end{proof}

Theorem~\ref{thm:pm-covariance} 
shows that,
when the d-PM is used with a finite number of samples $N$
and a finite number of AC iterations $\itrac$ 
to estimate the eigenvectors $\{ \mcoveigv_l\}_{l=1}^M$  of the true covariance matrix $\mcov$, 
then three different types of errors occur:
\begin{enumerate}
\item [E1:] Errors resulting from the finite number of AC iterations $\itrac$. 
     These errors are expressed in the matrix 
$\left( \sensorsel \avgmat{}^{\itrac}  \sensorsel^T \right)$. 
\item [E2:] Errors resulting from the finite number of samples $N$. 
     These errors are expressed in $\smcov$.
\item [E3:] Errors resulting from the finite number of PM iterations $\itrpm$,
which we neglect as stated in Assumption A2.
\end{enumerate}
Since the averaging matrix $\avgmat{}$ is assumed to satisfy the convergence condition
(\ref{eq:ac-convergence}), 
we conclude that
$K \sensorsel \avgmat{}^{\itrac}  \sensorsel^T \rightarrow \ones{M}\ones{M}$
as
$\itrac \rightarrow \infty$,
consequently, $\pdsmcov \rightarrow \smcov$.
We remark that
$\pdsmcov \rightarrow K \left( \sensorsel \avgmat{}^{\itrac}  \sensorsel^T \right) \odot \mcov$
when $N \rightarrow \infty $,
i.e.,
for a finite number of AC iterations $\itrac$,
the eigendecomposition of the sample covariance matrix
using the d-PM is not an asymptotically consistent estimator of 
the eigenvectors of the true measurement covariance matrix $\mcov$.

Theorem~\ref{thm:pm-covariance} simplifies the performance analysis of the d-PM,
since it provides 
a link to an equivalent centralized algorithm formulation
which can be analyzed using 
the conventional statistical analysis techniques and results \cite{brillinger2001time}.
In the sequel, we start our performance analysis by introducing 
the error vectors which represent E1 and E2 types of errors.
Then, we compute analytical expressions for these errors
and 
finally we derive the second order statistics of the eigenvectors 
obtained from the d-PM.    
 
For the centralized eigendecomposition, 
the sample estimate of the $l$th eigenvector $\smcoveigv_l$ of the
true covariance matrix $\mcov$ is expressed as
\begin{equation}
\smcoveigv_l \eq \mcoveigv_l + \error \mcoveigv_l,
\label{eq:error-eigenvector-cent}
\end{equation}
where the error vector $\error \mcoveigv_l$ accounts only for 
the finite sample effects,
i.e., E2 type of errors, used in the computation of the sample covariance matrix.
The decentralized estimate of the $l$th eigenvector is expressed as
\begin{equation}
\pdsmcoveigv{l} \eq \mcoveigv_l + \error \pdsmcoveigv{l},
\label{eq:error-eigenvector-dist}
\end{equation}
where the error vector $\error \pdsmcoveigv{l}$ accounts for 
errors resulting from the finite number of samples and 
the finite number of AC iterations,
i.e., E1 and E2 type of  errors.
Similarly, we define
\begin{equation}
\smcov \eq \mcov + \error \mcov
\label{eq:error-mcov-cent}
\end{equation}
and
\begin{equation}
\pdsmcov \eq \mcov + \error \pdsmcov,
\label{eq:error-mcov-dist}
\end{equation}
where $\error \mcov$ accounts only for E2 type of errors and 
$\error \pdsmcov$ accounts for both E1 and E2 types of errors.
Using the aforementioned notation,
the second order statistics of the eigenvectors computed using the d-PM
are expressed as
$\sosone$ 
and
$\sostwo$.

The following theorem simplifies the computation of 
$\sosone$
and
$\sostwo$
by introducing a simple expression for $\error \pdsmcoveigv{l}$.    

\begin{thm}
Under Assumptions A1 and A2,
the error vector $\error \pdsmcoveigv{l}$ 
is given by
$$
\error \pdsmcoveigv{l} = -\vamone_l \left( \error \mcov \, \mcoveigv_l +
\pvavone{l} \right) $$
where $\vamone_l$
is defined in Eq.~(\ref{eq:def-varmone}),
\begin{equation}
\pvavone{l} = K \sum_{k=2}^{K} \avgmateige_k^{\itrac} \, 
	      \vamthree_k \mcov \vamthree_k^H \mcoveigv_l,
\label{eq:def-vavone}
\end{equation}
$\vamthree_k=\diag{\sensorsel \pmb{\beta}_k }$,
$\sensorsel$ is defined in Eq.~(\ref{eq:def-t})
and $\avgmateigv_k$ and $\avgmateige_k$ are defined in Eq.~(\ref{eq:ac-eig}).

\label{thm:delta-eigenvector}
\end{thm}
\begin{proof}
See Appendix~\ref{apx:theorem2}.
\end{proof}

Note that in Theorem~\ref{thm:delta-eigenvector} the E1 type of errors are expressed in terms of 
the vector $\pvavone{l}$ which depends on the number of AC iterations $\itrac$ 
and on the eigenvalues and eigenvectors of the 
weighting matrix $\avgmat{}$, except for the principle eigenvalue and eigenvector.
Since the magnitude of $\avgmateige_k$ is strictly less than one
for $k=2 \mdots K$ (see Sec.~\ref{sec:ac}),
it follows from Eq.~(\ref{eq:def-vavone}) that
$\pvavone{l} \rightarrow \zeros$ as $\itrac \rightarrow \infty$,
i.e., the AC protocol is carried out for an infinitely large number of iterations.
Consequently, 
$\error \pdsmcoveigv{l}$ contains no E1 type of errors
when $\itrac \rightarrow \infty$.
In Theorem~\ref{thm:delta-eigenvector}, 
the E2 errors are expressed in terms of the matrix $\error \mcov$.
If an infinite number of sample is available,
i.e., $N \rightarrow\infty$ then $\error \mcov \rightarrow \zeros$.
Consequently, $\error \pdsmcoveigv{l} \rightarrow \zeros$ as both
$\itrac$ and $N$ tend to infinity.

Based on Theorem~\ref{thm:delta-eigenvector},
analytical expressions for 
$\sosone$
and
$\sostwo$
are introduced in the following theorem.
These expressions are useful for computing the MSE of estimators which
are based on the d-PM as we show later for the d-ESPRIT algorithm.
  
\begin{thm}
Under Assumptions A1 and A2  
$$
\begin{aligned}
\sosone
=& \dfrac{\mcoveige_l}{N}
\sum_{\underset{\scriptstyle{i\neq l}}{i=1}}^{M}
\dfrac{\mcoveige_i}{\left( \mcoveige_l - \mcoveige_i \right)^2} 
 \mcoveigv_i \mcoveigv_i^H 
 \chrod{l,m} \\
 & +
\vamone_l \pvavone{l} \pvavonet{m}{H} \vamone_m^H
\end{aligned}
$$
and
$$
\begin{aligned}
\sostwo =& 
\dfrac{\mcoveige_l \mcoveige_m}{N}
\dfrac{\mcoveigv_l \mcoveigv_m^T}{\left( \mcoveige_l -\mcoveige_m \right)^2}  (\chrod{l,m} - 1) 
\\ & +
\vamone_l \pvavone{l} \pvavonet{m}{T} \vamone_m^T
\end{aligned}
$$
where
$\chrod{l,m}$ is the Kronecker delta,
$N$ is the number of samples,
$\vamone_l$ is defined in Eq.~(\ref{eq:def-varmone})
and $\pvavone{l}$
is defined in Eq.~(\ref{eq:def-vavone}).
\label{thm:mse}
\end{thm}

\begin{proof}
See Appendix \ref{apx:theorem3}.
\end{proof}

Note that only the second terms in 
$\sosone$
and
$\sostwo$ 
depend on the number of AC iterations $\itrac$
(through the vectors $\pvavone{l}$ and $\pvavone{m}$)
and as $\itrac\!\rightarrow\!\infty$, these terms converge to zero.
Consequently, 
$\sosone$
and
$\sostwo$
tend to the centralized case found in
\cite{brillinger2001time},
when $\itrac \rightarrow \infty$.
Moreover, 
as $N \rightarrow \infty$
for $\itrac < \infty$,
$\sosone$ and $\sostwo$
do not converge to zero,
i.e. 
the d-PM is not a consistent estimator for $\{ \mcoveigv_l \}_{l=1}^M$,
unless $\itrac$ is infinitely large. 
Theorem~\ref{thm:mse}
shows that, 
in the second order statistics of the eigenvector estimates,
the AC errors appear as an additive error term
which is remarkable since in Theorem~\ref{thm:pm-covariance}
the corresponding errors for the sample covariance 
are expressed as an element-wise multiplication with the matrix 
$\left( \sensorsel \avgmat{}^{\itrac}  \sensorsel^T \right)$.
 
Note that in practice $\itrac$ can not be chosen to be arbitrarily large,
thus, the second terms in 
$\sosone$ and $\sostwo$
will always be non-zero.
However,
$\itrac$ can usually be chosen such that the second terms in  
$\sosone$ and $\sostwo$ 
are of the same order as the first terms.
The proper choice of $\itrac$ will be further addressed 
in the simulations in Sec.~\ref{sec:simualtion}.

\section{Performance Analysis of the d-ESPRIT Algorithm}
\label{sec:d-ESPRIT}
In this section, we briefly review   
the decentralized ESPRIT (d-ESPRIT) algorithm presented in \cite{suleiman2013decentralized}.
Then, results from Theorem~\ref{thm:mse} are  
applied to derive an analytical expression for the
MSE of DOA estimation using the d-ESPRIT algorithm.

\subsection{Signal Model and the ESPRIT Algorithm}
Consider a planar sensor array composed of $K$ identically oriented
uniform linear subarrays,
where the $k$th subarray is composed of $M_k$ sensors.
The distance $\sepdist$ between two successive sensors 
measured in half wavelength is 
identical for all subarrays, see Fig.~\ref{fig:topology}.
The displacements between the subarrays are considered to be unknown,
thus, the array is partly calibrated.

Signals of $L$ far-field and narrow-band sources impinge on the system of the subarrays from directions
$\doas=[\doa_1   \mdots   \doa_L]^T$.
The output of the $k$th subarray at time $t$ is given by
\begin{equation} 
\msr_k(t) = \steermat_k(\doas) \sourcesigs(t) + \noise_k(t),
\label{eq:def-subarray-response}
\end{equation} 
where $\msr_k(t) \triangleq [  \msrscal_{k,1}(t)   \mdots    \msrscal_{k, M_k}(t) ]^T$ 
is the baseband output of the $M_k$ sensors,
$\sourcesigs(t) \triangleq [\sourcesig_1(t)   \mdots   \sourcesig_L(t)]^T$ 
is the signal vector of the $L$ Gaussian sources and
$\noise_k(t) \triangleq [\noisescal_{k,1}(t)   \mdots   \noisescal_{k, M_k}(t)]^T$ 
is the vector of temporally and spatially complex circular white Gaussian sensor noise. 
The steering matrix of the $k$th subarray is
$\steermat_k(\doas) \triangleq [\steervec_k(\doa_1)   \mdots   \steervec_k(\doa_L)]$,
where $\steervec_k(\doa_l)$ is the response of the 
$k$th subarray corresponding to a source from direction $\doa_l$ 
relative to the array broadside
\begin{equation}
\steervec_k(\theta_l) = e^{ \jmath \pi \refsenloc_{k}^T \sincosdoa_l  } 
\left[ 1, e^{ \jmath \pi \sepdist \sin \doa_l }    
\mdots
e^{ \jmath (M_k - 1)  \pi \sepdist \sin \doa_l }  \right]^T,
\label{eq:svk}
\end{equation}
where $\sincosdoa_l = [\sin \doa_l, \cos \doa_l]^T$
and $\refsenloc_{k}$ is the position of the first sensor in the $k$th subarray
relative to the reference sensor in the first subarray,
which is considered to be unknown.
The measurement vector of the array is
\begin{equation}
\msr(t) = \steermat(\doas) \sourcesigs(t) + \noise(t),
\label{eq:def-array-response}
\end{equation} 
where 
$\msr(t) \triangleq [\msr_1^T(t) \mdots \msr_K^T(t)]^T$,
as in Eq.~(\ref{eq:measurements}),
$\steermat(\doas) \triangleq [\steermat_1^T(\doas) \mdots \steermat_K^T(\doas)]^T$ and 
$\noise(t) \triangleq [\noise_1^T(t) \mdots \noise_K^T(t)]^T$.
The measurement covariance matrix is
\begin{equation} 
\mcov \triangleq
\expect{ \msr(t) \msr^H(t)}  
  = \steermat(\doas) \scov \steermat^H(\doas) + \noisevariance \id_{M}
\label{eq:def-mcovk}
\end{equation} 
where $\scov \triangleq \expect{\sourcesigs(t) \sourcesigs^H(t) }$ is the source covariance matrix
and $\noisevariance$ is the noise variance
and $\expect{\noise(t) \noise^H(t) } = \noisevariance \id_{M}$.
The eigendecomposition of 
the measurement covariance matrix can be partitioned as
\begin{equation}
 \mcov =  
			\sigsub{} \sigeigs{} \sigsub{}^H + 
			 \noisesub{} \noiseeigs{} \noisesub{}^H
\end{equation}
where $\sigeigs{} \in \reals^{L \times L}$ and 
$\noiseeigs{}  \in \reals^{(M - L) \times (M - L)}$ 
are diagonal matrices containing the signal and noise eigenvalues, respectively, 
$\sigsub{} \!\triangleq\! \left[ \mcoveigv_1,\ldots,\mcoveigv_L \right] \!\in\! \complex^{M \times L}$ and 
$\noisesub{} \! \triangleq \!\left[ \mcoveigv_{L+1},\ldots, \mcoveigv_{M} \right]\! \in\! \complex^{M \times (M-L)}$ 
are the signal and noise eigenvector matrices, respectively, and
$\mcoveigv_{1},\ldots,\mcoveigv_{M}$
are the eigenvectors of the matrix $\mcov$
corresponding to the eigenvalues
$\mcoveige_{1}\!\geq \!\ldots \!\geq \!\mcoveige_{L}\!\geq\! \mcoveige_{L+1}\! =\! \ldots\! =\! \mcoveige_{M} \!= \!\noisevariance$,
see Eq.~(\ref{eq:cov-eig-def}).

The ESPRIT algorithm exploits the translational invariance structure of the measurement setup.
This invariance structure is expressed in Fig.~\ref{fig:topology},
where the sensors of the array are partitioned into upper and lower groups.
The upper group contains the first $M_k-1$ sensors of the $k$th subarray
and the lower group contains the last $M_k-1$ sensors of the $k$th subarray.
The two groups are related by a translation with lag $d$ \cite{Roy1989}.
Let $\uppersel \eq \diag { \uppersel_1 \mdots \uppersel_K }$
and $\lowersel \eq \diag { \lowersel_1 \mdots \lowersel_K }$,
where  
$\uppersel_k \eq \left[ \id_{M_k}, \zeros_{M_k} \right]$ and 
$\lowersel  \eq \left[ \zeros_{M_k-1}, \id_{M_k-1} \right]$,
denote the selection matrices corresponding to the upper and lower groups, respectively.
Based on the selection matrices, we define two matrices 
$
	\uppersigsub{} \triangleq \uppersel \sigsub{}
$ and $
	\lowersigsub{} \triangleq \lowersel \sigsub{}
$.
In the conventional Least Squares ESPRIT \cite{Roy1989}, the DOAs 
are computed from the eigenvalues of the matrix
\begin{equation}
\espritdelay = \left(\uppersigsub{}^H \uppersigsub{} \right)^{-1} \uppersigsub{}^H \lowersigsub{}
\label{eq:esprit-compute-delay}
\end{equation}  
as follows
\begin{equation}
\doa_l = \sin^{-1}(\arg( \espritdelayeig_l )/(\pi \sepdist))
\label{eq:esprit-doas}
\end{equation}
where $\espritdelayeig_l$ for $l=1 \mdots L$ are the eigenvalues of the matrix $\espritdelay$,
see \cite{Roy1989} for details.   

In practice, the true covariance matrix (\ref{eq:def-mcovk}) is not available and its finite sample estimate
$\smcov$ is calculated from $N$ snapshot of the array output as in Eq.~(\ref{eq:def-smcov}).
Let $\ssigsub{}, \suppersigsub{}, \slowersigsub{}, \smcoveigv_i, \sespritdelay$, and $\sespritdelayeig_l$ 
be the estimates of 
$\sigsub{}, \uppersigsub{}, \lowersigsub{}, \mcoveigv_i, \espritdelay$, and $\espritdelayeig_l$, respectively,
obtained from the eigendecomposition of the sample covariance matrix $\smcov$.
\begin{figure}
\centering
\includegraphics[width=0.57\columnwidth]{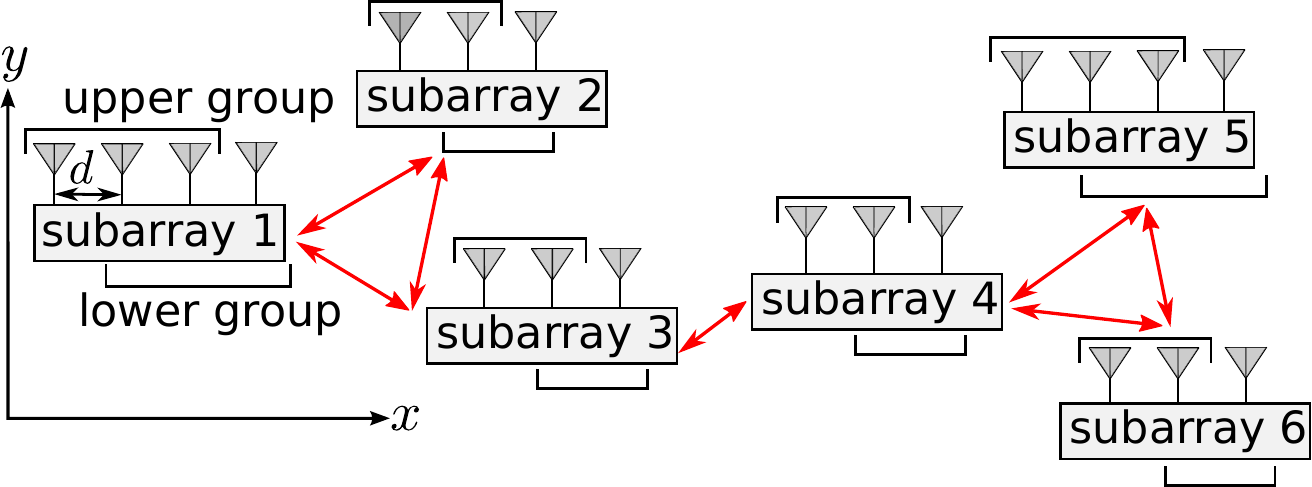}
\caption{Array topology and sensor grouping for 6 subarrays. 
The read lines indicate communication links between neighboring subarrays.}
\label{fig:topology} 
\end{figure}

\subsection{The d-ESPRIT Algorithm}
The decentralized ESPRIT (d-ESPRIT) algorithm which is presented in \cite{suleiman2013decentralized} 
comprises two decentralized steps,
first, the decentralized signal subspace estimation using the d-PM,
second, the decentralized estimation of the matrix $\espritdelay$.

The decentralized signal subspace estimation is carried out as explained in Sec.~\ref{sec:dpm}.
The resulting decentralized estimate of $\sigsub{}$, denoted as $\dssigsub{}$,
is distributed among the subarrays,
where each subarray stores only a part from $\dssigsub{}$.
In \cite{suleiman2013decentralized}, based on the AC protocol,
a decentralized algorithm for estimating the matrix $\espritdelay$ is introduced.
Denote the corresponding estimate at the $k$th subarray
as $\dsespritdelay_{\cpy{k}}$.
In \cite{suleiman2013decentralized},
the computation of $\dsespritdelay_{\cpy{k}}$ is achieved 
by rewriting Eq.~(\ref{eq:esprit-compute-delay}) 
as
\begin{equation}
\espritdelayt_{\cpy{k}} \dsespritdelay_{\cpy{k}} = \espritdelaytt_{\cpy{k}},
\label{eq:estimate-psi}
\end{equation}
where 
$\espritdelayt_{\cpy{k}}$
and
$\espritdelaytt_{\cpy{k}}$
are respectively the decentralized estimate of
$\uppersigsub{}^H \uppersigsub{}$
and
$\uppersigsub{}^H \lowersigsub{}$,
at the $k$th subarray.
The AC protocol is used to compute each entry of the matrices
$\espritdelayt_{\cpy{k}}$ and $\espritdelaytt_{\cpy{k}}$ as follows:
\begin{equation}
\espritdelayst_{i, j, \cpy{k}} \!\!= \!\!K \!\! 
	\left[  \!
	   \avgmat{}\!^{\itracttt} \!
	   [\dsmcoveigv_{i, 1}^H \uppersel_1^H \!  \uppersel_1 \dsmcoveigv_{j, 1} \mdots 
	   \dsmcoveigv_{i, K}^H \uppersel_K^H  \uppersel_K \dsmcoveigv_{j, K}]^T\! 
    \right]_{k}\!,
   \label{eq:d-esprit-psi-1}
\end{equation}
and
\begin{equation}
\espritdelaystt_{i, j, \cpy{k}} \!\!= \!\!K \!\! 
	\left[  \!
	   \avgmat{}\!^{\itracttt} \! 
	   [\dsmcoveigv_{i, 1}^H \uppersel_1^H  \lowersel_1 \! \dsmcoveigv_{j, 1} \mdots 
	   \dsmcoveigv_{i, K}^H \uppersel_K^H  \lowersel_K \dsmcoveigv_{j, K}]^T\! 
    \right]_{k}\!,
    \label{eq:d-esprit-psi-2}
\end{equation}
where 
$\espritdelayst_{i, j, \cpy{k}}$
and
$\espritdelaystt_{i, j, \cpy{k}}$
denote the $(i,j)$th entries of the matrices
$\espritdelayt_{\cpy{k}}$
and
$\espritdelaytt_{\cpy{k}}$,
respectively,
and
$\itracttt$ denotes
the number of AC iterations used to compute 
$\espritdelayst_{i, j, \cpy{k}}$
and
$\espritdelaystt_{i, j, \cpy{k}}$.
Thus, the $k$th subarray can estimate the matrix $\espritdelay$
as presented in Eq.~(\ref{eq:estimate-psi}).
The DOA estimation is carried out locally in the $k$th 
subarray using the eigenvalues of the matrix $\dsespritdelay_{\cpy{k}}$
in Eq.~(\ref{eq:esprit-doas}).

%
%
%
To simplify the performance analysis of the d-ESPRIT algorithm, 
the following assumption is made
\begin{enumerate}
  \item [A3:] The number of AC iterations $\itracttt$ which is used to compute
  $\dsespritdelay_{\cpy{k}}$ from $\dssigsub{}$ is large compared to the number
  of AC iterations used to compute $\dssigsub{}$, i.e., $\itracttt \gg \itrac$.  
\end{enumerate}
Under Assumption A3, 
the AC errors
in 
the decentralized estimate of $\espritdelay$
are negligible compared to those in $\dssigsub{}$.
Thus,
the decentralized estimate of $\espritdelay$
is a good approximation of the centralized 
one
using $\dssigsub{}$,
i.e., 
\begin{equation}
\dsespritdelay_{\cpy{k}} = \dsespritdelay =  \left(\dsuppersigsub{}^H \dsuppersigsub{} \right)^{-1}
    \dsuppersigsub{}^H \dslowersigsub{},
\label{eq:esprit-dist-delay-dist}
\end{equation}
where $\dsuppersigsub{} \eq \uppersel \dssigsub{}$,
 $\dslowersigsub{} \eq \lowersel \dssigsub{}$
 and $k=1 \mdots K$.
Let $\dsespritdelayeig_l$ for $l=1 \mdots L$
be the eigenvalues of $\dsespritdelay$.
In the d-ESPRIT algorithm, $\dsespritdelayeig_l$ is used as an estimate of
$\espritdelayeig_l$. Thus, we define the estimation error 
$\error \dsespritdelayeig_l$ as
\begin{equation}
\dsespritdelayeig_l \eq \espritdelayeig_l + \error \dsespritdelayeig_l,
\label{eq:d-esprit-dsespritdelayeig}
\end{equation}
for $l=1 \mdots L$,
where $\error \dsespritdelayeig_l$ accounts for both 
E1 and E2 types of error.

\subsection{The MSE for DOA Estimation based on d-ESPRIT}
\label{sec:MSE-ours}
In \cite{rao1989performance},
the MSE of DOA estimation using 
the conventional Least Squares ESPRIT is presented.
Assumption A3 allows us to use the analysis from \cite{rao1989performance}
by replacing 
$\expect{  \error \mcoveigv_l \, \error \mcoveigv_m^H }$
and $\expect{ \error \mcoveigv_l \, \error \mcoveigv_m^T }$
with 
$\sosone$
and $\sostwo$,
respectively.
Thus, the following result from \cite{rao1989performance} holds,
\begin{equation}
\expect{ (\error \dsdoa_l)^2 } = 
         \dfrac{ \expect{ | \error \dsespritdelayeig_l  |^2 } - 
         \mRe [ \big( \espritdelayeig_l^* \big)^2  
           \expect{  ( \error \dsespritdelayeig_l )^2} ]  }
           {2 \left( \pi \sepdist{} \cos \doa_l \right)^2},
\label{eq:rao-doa-error}  
\end{equation}
where
\begin{equation}
\expect { | \error \dsespritdelayeig_l  |^2 } =
		\rmsevone_l^H \expect {
			\error \dssigsub{} \,  
			\espritdelayeigrv _l \espritdelayeigrv ^H_l
			\error \dssigsub{}^H
		} \rmsevone_l, 
\label{eq:rao-1} 
\end{equation}
\begin{equation}
\begin{aligned}
\expect { ( \error \dsespritdelayeig_l  )^2 } =
	\rmsevtwo_l^H \expect {
		\error \dssigsub{} \, 
		\espritdelayeigrv_l \espritdelayeigrv^T_l
		\error \dssigsub{}^T
	} \rmsevtwo_l^*, 
\end{aligned}
\label{eq:rao-2}
\end{equation} 
$\dssigsub{} = \sigsub{} + \error \dssigsub{}$,
$
\rmsevone_l^H \eq 
	\espritdelayeiglv_l \left( \uppersigsub{}^H \uppersigsub{} \right)^{-1} 
	\uppersigsub{}^H 
	\left( \uppersel - \espritdelayeig_l^* \lowersel \right)
$
and
$
\rmsevtwo_l^H \eq 
	\espritdelayeiglv_l \left( \uppersigsub{}^H \uppersigsub{} \right)^{-1} 
	\uppersigsub{}^H 
	\left(\lowersel  - \espritdelayeig_l \uppersel \right)
$.
The $l$th left and right eigenvectors 
which correspond to the $l$th eigenvalue
of the matrix $\espritdelay$
are denoted as
$\espritdelayeiglv_l$ and $\espritdelayeigrv_l$,
respectively.
In Eq.~(\ref{eq:rao-1}) and Eq.~(\ref{eq:rao-2}),
we replaced
the conventional error $\error \sigsub{}$
by $\error \dssigsub{}$
in the corresponding expressions of \cite{rao1989performance}.

Using the expression from Theorem~\ref{thm:mse},
the expectation of the right hand side of Eq.~(\ref{eq:rao-1}) is rewritten 
as  
\begin{equation}
\begin{aligned}
\expect{
			\error \dssigsub{} \, 
			\espritdelayeigrv_l \espritdelayeigrv^H_l
			\error \dssigsub{}^H
		} &= \sum_{i=1}^{L}\sum_{j=1}^{L} 
		           \left[ \espritdelayeigrv_{l}\espritdelayeigrv_{l}^H \right]_{i,j} 
		           \sosonett \\
&= \frac{1}{N}\sum_{i=1}^{L} \sum_{\underset{\scriptstyle{j\neq i}}{j=1}}^{M}
  \dfrac{\mcoveige_i \mcoveige_j}{(\mcoveige_i - \mcoveige_j)^2} 
      \left[ \espritdelayeigrv_{l}\espritdelayeigrv_{l}^H \right]_{i,i}
          \mcoveigv_i  \mcoveigv_i^H \\
	&  +	\sum_{i=1}^{L}\sum_{j=1}^{L}
	\left[ \espritdelayeigrv_{l}\espritdelayeigrv_{l}^H \right]_{i,j}
	\vamone_i \pvavone{i} \pvavonet{j}{H} \vamone_j^H.
\end{aligned}
\label{eq:rmse-result-1} 
\end{equation}

Similarly, the expectation of the right hand side of Eq.~(\ref{eq:rao-2}) 
is written as
\begin{equation}
\begin{aligned}
\expect {
			\error \dssigsub{} \, 
			\espritdelayeigrv_l \espritdelayeigrv^T_l
			\error \dssigsub{}^T
		} &= \sum_{i=1}^{L}\sum_{j=1}^{L} 
		           \left[ \espritdelayeigrv_{l}\espritdelayeigrv_{l}^T \right]_{i,j} 
		           \sostwott \\
		  &=  -\sum_{i=1}^{L}\sum_{ \small{ \begin{matrix} j=1 \\ j \neq i \end{matrix} }}^{L}
		  \dfrac{ [ \espritdelayeigrv_{l}\espritdelayeigrv_{l}^T ]_{i,j} \mcoveige_i \mcoveige_j  \mcoveigv_i \mcoveigv_j^T}
		        {N \left( \mcoveige_i -\mcoveige_j \right)^2} \\ 
		& + \sum_{i=1}^{L}\sum_{j=1}^{L} [ \espritdelayeigrv_{l}\espritdelayeigrv_{l}^T ]_{i,j}
			\vamone_i \pvavone{i} \pvavonet{j}{T} \vamone_j^T.
\end{aligned} 
\label{eq:rmse-result-2}
\end{equation}

Equations (\ref{eq:rao-doa-error})--(\ref{eq:rmse-result-2})
provide the analytical expressions of the MSE for 
the DOA estimator using the d-ESPRIT algorithm.
The second terms of Equations (\ref{eq:rmse-result-1}) and 
(\ref{eq:rmse-result-2})
differ from the expressions in \cite{rao1989performance}.
Note that because of these terms, 
the MSE does not approach zero even if an infinitely large number of samples 
is available,
i.e. the d-ESPRIT algorithm is not a consistent estimator of the DOAs.
However, the simulations in Sec.~\ref{sec:simualtion}
demonstrate that for a finite number of samples and a moderate SNR, 
a finite number of AC iterations is sufficient to achieve 
a performance comparable to that of the conventional ESPRIT algorithm \cite{Roy1989},
and to achieve the CRB \cite{see2004direction}.

\section{Simulation Results}
\label{sec:simualtion}
An array composed of $K\!=\!6$ subarrays each containing 2 sensors, i.e., $M\!=\!12$,
separated by half-wavelength 
is used in the simulations.
The locations of the first sensors at the subarrays are 
$(0,      0)$,
$(0.45,   0.99)$,
$(3.02,    0.45)$,
$(5.61,    0.90)$,
$(8.03,    1.46)$
and
$(8.70,    0.50)$
measured in half-wavelength.
The upper and lower selection matrices of the $k$th subarray are 
$\uppersel_k=[1, 0]$ and $\lowersel_k=[0, 1]$.
The array topology depicted in Fig.~\ref{fig:topology}
is assumed. Thus,
the neighboring sets are 
$\neighbor{1}=\{2, 3\},
\neighbor{2}=\{1, 3\},
\neighbor{3}=\{1, 2, 4\},
\neighbor{4}=\{3, 5, 6\},
\neighbor{5}=\{4, 6\}
$
and
$
\neighbor{6}=\{4, 5\}$,
where the $k$th subarray communicates only with its neighbors $\neighbor{k}$.
The entries of the weighting matrix $\avgmat{}$ are selected as follows 
\begin{equation}
[\avgmat{}]_{i,j} = \left\{ \!\!\!\!
\begin{array}{l l}
\frac{1}{\max\{\card{ \neighbor{i} }, \card{ \neighbor{j}} \}}, \!\!\!   &  \text{ if }    i \neq j \text{ and }  j \in \neighbor{i}\\
\avgmat{i},   &   \text{ if }   i = j \\
0,       & \text{ otherwise},
\end{array} 
\right.
\label{eq:avg-choice-weighting-factor}
\end{equation}  
where $\card{ \neighbor{i} }$ is the number of elements in the set $\neighbor{i}$.
The weighting factors $\{\avgmat{i}\}_{i=1}^K$ are chosen as   
$\avgmat{i} = 1 - \sum_{j=1}^K \avgmat{i,j}$;
refer to \cite{Xiao2004} for further details.
This choice of the weighting factors only requires that 
each node knows the degree of its neighbors,
thus,
local but not global knowledge about the network topology 
is required at the node level.
The weighting matrix $\avgmat{}$ resulting from 
the weighting scheme in Eq.~(\ref{eq:avg-choice-weighting-factor})
guarantees asymptotic convergence of the AC protocol, 
provided that the graph associated with the network is not bipartite \cite{Xiao2004}. 

Signals from $L=3$ equal-powered Gaussian sources impinge onto the array from directions 
$-14^{\circ}, -10^{\circ}$ and $5^{\circ}$.
In the sequel, we evaluate our analytical expressions
for the performance of the d-PM and the d-ESPRIT algorithm. 

\subsection{Performance of the d-PM}
In the first set of simulations, shown in Fig.~\ref{fig:d-pm-snr} and Fig.~\ref{fig:d-pm-n},
the performance of the d-PM is evaluated as follows. 
We estimate $L=3$ eigenvectors of the sample covariance matrix at the $i$th realization using the d-PM,
i.e., we compute $\dssigsub{}(i)=[\dsmcoveigv_1(i), \dsmcoveigv_2(i), \dsmcoveigv_3(i) ]$,
for $i=1 \mdots 200$ realizations.
Since the eigendecomposition is unique up to a multiplication with a unity-magnitude complex scalar,
we use the method introduced in \cite[Eq.~(54)]{friedlander1998second} to compute this scalar
and correct the estimated eigenvectors.
Then, the error matrix $\error \dssigsub{}(i)=\dssigsub{}(i)-\sigsub{}$ is computed,
where $\sigsub{}=[\mcoveigv_1, \mcoveigv_2, \mcoveigv_3]$ is the true signal subspace.
At the $i$th realization, we define the normalized square error (SE) of the d-PM as 
\begin{equation}
\dpmse(i) \eq \trace{\error \dssigsub{}(i)\error \dssigsub{}^H(i)}/\trace{\sigsub{} \sigsub{}^H}.
\end{equation}
Finally, we compute the root mean square error (RMSE)
\begin{equation}
\vspace{-1mm}
\dpmrmse \eq  \left( \frac{1}{200} \sum_{i=1}^{200}\dpmse(i) \right)^{\frac{1}{2}}.
\vspace{-1mm}
\label{eq.dpmrmse}
\end{equation}
The RMSE which is obtained from our analytical expression
for the d-PM algorithm is denoted as
$\adpmrmse$ 
and it is computed as
\begin{equation}
\vspace{-1mm}
\adpmrmse \eq  \left( \sum_{l=1}^{3} \sosone / \trace{\sigsub{} \sigsub{}^H} \right)^{\frac{1}{2}},
\label{eq.dpmrmse}
\end{equation}
where 
$\sosone$
is given in Theorem~\ref{thm:mse}.

In Fig.~\ref{fig:d-pm-snr}, we compare $\dpmrmse$ from 200 realizations with the 
$\adpmrmse$ at different $\snr$s where the number of samples is fixed to $N=100$ 
and the number of AC iterations $\itrac$ is taken to be $10, 20$ and $30$.
The number of PM iterations is fixed to $\itrpm=10$ for all simulations.
It can be seen from Fig.~\ref{fig:d-pm-snr} that 
the error in the estimated eigenvectors decreases with increasing $\snr$
until it reaches a certain value, which depends on $\itrac$, then it is saturated.
Note that the error computed using 
the analytical expressions $\adpmrmse$ 
corresponds well to the one computed over 200 realizations $\dpmrmse$.

In Fig.~\ref{fig:d-pm-n}, the $\snr$ is set to $10$ dB and  
$\dpmrmse$ and $\adpmrmse$ are computed for different numbers of samples $N$
for three different numbers of AC iterations $10, 20$ and $30$.
The number of PM iterations is fixed to $\itrpm=10$.
From Fig.~\ref{fig:d-pm-n}, 
it can be observed that the  
error in the estimated eigenvectors decreases with $N$ for  small values of $N$.
However when $N$ is large, $\dpmrmse$ and $\adpmrmse$ do not change with $N$
as it can be seen in Fig.~\ref{fig:d-pm-n} 
for $\itrac=10$ and $\itrac=20$.
For $\itrac=30$, $\dpmrmse$ and $\adpmrmse$ 
show a good correspondence at very large values of $N$
(which is not displayed in the figure).   
Moreover, a larger number of AC iterations results in a smaller error.
This behaviour of the  $\dpmrmse$
is in accordance with our conclusion that the 
d-PM is a consistent estimator of the eigenvectors of the true measurement covariance matrix
only when $\itrac$ is infinitely large, see Sec.~\ref{sec:d-PM-analysis}.
It can also be observed in Fig.~\ref{fig:d-pm-n} that 
the error computed using the analytical expressions $\adpmrmse$ 
corresponds to the one computed from 200 realizations of $\dpmrmse$.

\subsection{Performance of the d-ESPRIT Algorithm}
In the second set of simulations, whose results are shown in Fig.~\ref{fig:d-esprit-snr} and 
Fig.~\ref{fig:d-esprit-n},
the performance of the d-ESPRIT algorithm is evaluated and compared to the analytical expressions of Sec.~\ref{sec:d-ESPRIT}.                
In these simulations, we assume that the number of sources $L=3$
is known to all subarrays,
which is the case in many applications, 
e.g. communications, condition monitoring and acoustics.
In applications where $L$ is not known eigenvalue-based detection
criteria available in the literature 
(such as the MDL \cite{wax1985detection} and the approaches in \cite{lu2013flexible} and \cite{brcich2002detection}) 
can be adapted in the decentralized scenario to detect the number of sources.
The RSME of the d-ESPRIT algorithm is computed over 200 realizations as
\begin{equation}
\despritrmse \eq \left( \frac{1}{200} \sum_{i=1}^{200} \frac{1}{3}\sum_{l=1}^{3}
(\dsdoa_l(i) - \doa_l)^2  
\right)^{\frac{1}{2}},
\label{eq:rmse}
\end{equation}
where $\dsdoa_l(i)$ is the estimate of $\doa_l$ computed at the $i$th realization using the d-ESPRIT algorithm.   
The analytical expression of the RMSE of the d-ESPRIT algorithm is
\begin{equation}
\vspace{-1mm}
\adespritrmse \eq \left( \frac{1}{3}\sum_{l=1}^{3}
\expect{(\error \dsdoa_l)^2}  
\right)^{\frac{1}{2}},
\label{eq:sim-d-esprit-rmse}
\vspace{-1mm}
\end{equation}  
where $\expect{(\error \dsdoa_l)^2}$ is computed using Equations (\ref{eq:rao-doa-error})--(\ref{eq:rmse-result-2}).  
In this set of simulations, the number of PM iterations is set to $\itrpm=2$.
Moreover, in this simulation we plot the RMSE of 
the conventional ESPRIT algorithm computed as in Eq.~(\ref{eq:rmse}),
along with the corresponding 
performance analysis of 
the conventional ESPRIT algorithm from \cite{rao1989performance},
which we denote as $\aespritrmse$,  and 
the CRB for the conventional partly calibrated arrays \cite{see2004direction}.    
  
Fig.~\ref{fig:d-esprit-snr} demonstrates $\adespritrmse$ and $\despritrmse$ 
for different $\snr$s where a fixed number of samples $N=100$ is assumed. 
Note that at low $\snr$s the performance of the d-ESPRIT algorithm is similar to 
that of the conventional ESPRIT algorithm
and it improves with increasing $\snr$.
However, at high $\snr$s,
it can be observed that the performance of the d-ESPRIT algorithm 
deviates from that of the conventional ESPRIT algorithm.
It is clear from Fig.~\ref{fig:d-esprit-snr} that
this deviation depends on the number of AC iterations $\itrac$.
Thus, for $\itrac=30$ and $\snr$ values up to $\snr = 15$ dB
the performance of the d-ESPRIT algorithm is similar to that of the conventional 
ESPRIT algorithm and both achieve the conventional CRB,
whereas for $\itrac=10$ this deviation starts from $\snr=0$ dB.
Moreover, it can be seen from Fig.~\ref{fig:d-esprit-snr} that
the RMSE of the d-ESPRIT algorithm at high $\snr$s is saturated and cannot be decreased unless
the number of AC iterations is increased.

\begin{figure}[t]
\centering
\definecolor{mycolor1}{rgb}{0.00000,0.49804,0.00000}%
\begin{tikzpicture}

\begin{axis}[%
width=\mfw\columnwidth,
scale only axis,
separate axis lines,
every outer x axis line/.append style={black},
every x tick label/.append style={font=\scriptsize\color{black}},
xmin=-30.8817163697987,
xmax=71.2094070622041,
xlabel={\SNR},
xmajorgrids,
every outer y axis line/.append style={black},
every y tick label/.append style={font=\scriptsize\color{black}},
ymode=log,
ymin=0.0961804182489373,
ymax=2.73976460447577,
yminorticks=true,
ylabel={\NormalizedRMSE},
ylabel style={yshift=-10pt},
ymajorgrids,
yminorgrids,
legend style={at={(0.50461,0.581233)},anchor=south west,legend cell align=left,align=left,fill=white}
]
\addplot [color=blue,line width=\mlw pt, mark repeat = 2, mark phase = {2}, only marks,mark=square,mark options={solid}]
  table[row sep=crcr]{%
-30.000	1.221\\
-28.000	1.212\\
-26.000	1.222\\
-24.000	1.201\\
-22.000	1.189\\
-20.000	1.165\\
-18.000	1.125\\
-16.000	1.068\\
-14.000	0.967\\
-12.000	0.819\\
-10.000	0.662\\
-8.000	0.529\\
-6.000	0.450\\
-4.000	0.408\\
-2.000	0.390\\
0.000	0.393\\
2.000	0.382\\
4.000	0.375\\
6.000	0.371\\
8.000	0.367\\
10.000	0.366\\
12.000	0.365\\
14.000	0.365\\
16.000	0.365\\
18.000	0.365\\
20.000	0.365\\
22.000	0.365\\
24.000	0.366\\
26.000	0.364\\
28.000	0.363\\
30.000	0.365\\
32.000	0.365\\
34.000	0.365\\
36.000	0.365\\
38.000	0.365\\
40.000	0.365\\
42.000	0.366\\
44.000	0.364\\
46.000	0.366\\
48.000	0.367\\
50.000	0.367\\
52.000	0.367\\
54.000	0.368\\
56.000	0.369\\
58.000	0.368\\
60.000	0.366\\
62.000	0.366\\
64.000	0.365\\
66.000	0.364\\
68.000	0.362\\
70.000	0.361\\
72.000	0.361\\
74.000	0.363\\
76.000	0.361\\
78.000	0.361\\
80.000	0.377\\
};
\addlegendentry{\dPMT};

\addplot [color=green, line width=\mlw pt, mark repeat = 2, mark phase = {1}, only marks,mark=o,mark options={solid}]
  table[row sep=crcr]{%
-30.000	1.214\\
-28.000	1.209\\
-26.000	1.219\\
-24.000	1.199\\
-22.000	1.187\\
-20.000	1.167\\
-18.000	1.120\\
-16.000	1.059\\
-14.000	0.967\\
-12.000	0.812\\
-10.000	0.666\\
-8.000	0.516\\
-6.000	0.368\\
-4.000	0.270\\
-2.000	0.220\\
0.000	0.188\\
2.000	0.178\\
4.000	0.166\\
6.000	0.159\\
8.000	0.154\\
10.000	0.151\\
12.000	0.147\\
14.000	0.143\\
16.000	0.139\\
18.000	0.138\\
20.000	0.138\\
22.000	0.137\\
24.000	0.137\\
26.000	0.137\\
28.000	0.137\\
30.000	0.136\\
32.000	0.136\\
34.000	0.136\\
36.000	0.136\\
38.000	0.136\\
40.000	0.136\\
42.000	0.135\\
44.000	0.136\\
46.000	0.136\\
48.000	0.137\\
50.000	0.137\\
52.000	0.136\\
54.000	0.136\\
56.000	0.137\\
58.000	0.137\\
60.000	0.137\\
62.000	0.136\\
64.000	0.136\\
66.000	0.136\\
68.000	0.136\\
70.000	0.136\\
72.000	0.135\\
74.000	0.135\\
76.000	0.131\\
78.000	0.128\\
80.000	0.134\\
};
\addlegendentry{\dPMTT};
\addplot [color=red,line width=\mlw pt, only marks,mark repeat = 2, mark phase = {2}, mark=x,mark options={solid}]
  table[row sep=crcr]{%
-30.000	1.215\\
-28.000	1.202\\
-26.000	1.217\\
-24.000	1.195\\
-22.000	1.184\\
-20.000	1.164\\
-18.000	1.117\\
-16.000	1.060\\
-14.000	0.964\\
-12.000	0.829\\
-10.000	0.689\\
-8.000	0.562\\
-6.000	0.410\\
-4.000	0.294\\
-2.000	0.232\\
0.000	0.190\\
2.000	0.173\\
4.000	0.156\\
6.000	0.145\\
8.000	0.140\\
10.000	0.135\\
12.000	0.131\\
14.000	0.127\\
16.000	0.124\\
18.000	0.121\\
20.000	0.120\\
22.000	0.120\\
24.000	0.119\\
26.000	0.120\\
28.000	0.119\\
30.000	0.118\\
32.000	0.118\\
34.000	0.118\\
36.000	0.118\\
38.000	0.119\\
40.000	0.118\\
42.000	0.117\\
44.000	0.118\\
46.000	0.118\\
48.000	0.119\\
50.000	0.119\\
52.000	0.118\\
54.000	0.118\\
56.000	0.118\\
58.000	0.119\\
60.000	0.119\\
62.000	0.118\\
64.000	0.118\\
66.000	0.118\\
68.000	0.118\\
70.000	0.119\\
72.000	0.118\\
74.000	0.117\\
76.000	0.113\\
78.000	0.112\\
80.000	0.115\\
};
\addlegendentry{\dPMTTT};

\addplot [color=black, line width=0.1 pt, mark repeat=2, mark phase = {1} ,only marks,mark=*,mark options={solid}]
  table[row sep=crcr]{%
-30.000	1.219\\
-28.000	1.207\\
-26.000	1.217\\
-24.000	1.197\\
-22.000	1.186\\
-20.000	1.156\\
-18.000	1.130\\
-16.000	1.051\\
-14.000	0.959\\
-12.000	0.823\\
-10.000	0.693\\
-8.000	0.577\\
-6.000	0.436\\
-4.000	0.309\\
-2.000	0.238\\
0.000	0.194\\
2.000	0.175\\
4.000	0.157\\
6.000	0.146\\
8.000	0.140\\
10.000	0.134\\
12.000	0.129\\
14.000	0.125\\
16.000	0.121\\
18.000	0.117\\
20.000	0.114\\
22.000	0.114\\
24.000	0.112\\
26.000	0.112\\
28.000	0.111\\
30.000	0.111\\
32.000	0.110\\
34.000	0.110\\
36.000	0.110\\
38.000	0.111\\
40.000	0.111\\
42.000	0.110\\
44.000	0.110\\
46.000	0.111\\
48.000	0.111\\
50.000	0.111\\
52.000	0.111\\
54.000	0.110\\
56.000	0.111\\
58.000	0.111\\
60.000	0.111\\
62.000	0.110\\
64.000	0.110\\
66.000	0.110\\
68.000	0.111\\
70.000	0.111\\
72.000	0.110\\
74.000	0.110\\
76.000	0.111\\
78.000	0.112\\
80.000	0.114\\
};
\addlegendentry{\dPMCent};

\addplot [color=blue,dashed, line width=\mlw pt, ]
  table[row sep=crcr]{%
-30.000	2.434\\
-28.000	1.975\\
-26.000	1.612\\
-24.000	1.326\\
-22.000	1.101\\
-20.000	0.925\\
-18.000	0.788\\
-16.000	0.682\\
-14.000	0.600\\
-12.000	0.538\\
-10.000	0.490\\
-8.000	0.455\\
-6.000	0.428\\
-4.000	0.407\\
-2.000	0.392\\
0.000	0.380\\
2.000	0.371\\
4.000	0.365\\
6.000	0.360\\
8.000	0.356\\
10.000	0.352\\
12.000	0.350\\
14.000	0.348\\
16.000	0.347\\
18.000	0.345\\
20.000	0.345\\
22.000	0.344\\
24.000	0.343\\
26.000	0.343\\
28.000	0.342\\
30.000	0.342\\
32.000	0.342\\
34.000	0.342\\
36.000	0.342\\
38.000	0.342\\
40.000	0.341\\
42.000	0.341\\
44.000	0.341\\
46.000	0.341\\
48.000	0.341\\
50.000	0.341\\
52.000	0.341\\
54.000	0.341\\
56.000	0.341\\
58.000	0.341\\
60.000	0.341\\
62.000	0.341\\
64.000	0.341\\
66.000	0.341\\
68.000	0.341\\
70.000	0.341\\
72.000	0.341\\
74.000	0.341\\
76.000	0.341\\
78.000	0.341\\
80.000	0.341\\
};
\addlegendentry{\AdPMT};

\addplot [color=green,dash pattern=on 1pt off 3pt on 4pt off 1pt,line width=\mlw pt]
  table[row sep=crcr]{%
-30.000	2.347\\
-28.000	1.893\\
-26.000	1.531\\
-24.000	1.244\\
-22.000	1.016\\
-20.000	0.835\\
-18.000	0.691\\
-16.000	0.576\\
-14.000	0.485\\
-12.000	0.413\\
-10.000	0.356\\
-8.000	0.310\\
-6.000	0.273\\
-4.000	0.244\\
-2.000	0.221\\
0.000	0.203\\
2.000	0.188\\
4.000	0.177\\
6.000	0.167\\
8.000	0.160\\
10.000	0.154\\
12.000	0.149\\
14.000	0.145\\
16.000	0.143\\
18.000	0.140\\
20.000	0.138\\
22.000	0.137\\
24.000	0.136\\
26.000	0.135\\
28.000	0.134\\
30.000	0.133\\
32.000	0.133\\
34.000	0.132\\
36.000	0.132\\
38.000	0.132\\
40.000	0.132\\
42.000	0.132\\
44.000	0.131\\
46.000	0.131\\
48.000	0.131\\
50.000	0.131\\
52.000	0.131\\
54.000	0.131\\
56.000	0.131\\
58.000	0.131\\
60.000	0.131\\
62.000	0.131\\
64.000	0.131\\
66.000	0.131\\
68.000	0.131\\
70.000	0.131\\
72.000	0.131\\
74.000	0.131\\
76.000	0.131\\
78.000	0.131\\
80.000	0.131\\
};
\addlegendentry{\AdPMTT};

\addplot [color=red,solid,line width=\mlw pt]
  table[row sep=crcr]{%
-30.000	2.344\\
-28.000	1.889\\
-26.000	1.528\\
-24.000	1.241\\
-22.000	1.012\\
-20.000	0.831\\
-18.000	0.686\\
-16.000	0.572\\
-14.000	0.480\\
-12.000	0.407\\
-10.000	0.349\\
-8.000	0.302\\
-6.000	0.265\\
-4.000	0.235\\
-2.000	0.211\\
0.000	0.192\\
2.000	0.176\\
4.000	0.164\\
6.000	0.154\\
8.000	0.146\\
10.000	0.139\\
12.000	0.134\\
14.000	0.130\\
16.000	0.127\\
18.000	0.124\\
20.000	0.122\\
22.000	0.120\\
24.000	0.119\\
26.000	0.118\\
28.000	0.117\\
30.000	0.116\\
32.000	0.116\\
34.000	0.115\\
36.000	0.115\\
38.000	0.115\\
40.000	0.115\\
42.000	0.114\\
44.000	0.114\\
46.000	0.114\\
48.000	0.114\\
50.000	0.114\\
52.000	0.114\\
54.000	0.114\\
56.000	0.114\\
58.000	0.114\\
60.000	0.114\\
62.000	0.114\\
64.000	0.114\\
66.000	0.114\\
68.000	0.114\\
70.000	0.114\\
72.000	0.114\\
74.000	0.114\\
76.000	0.114\\
78.000	0.114\\
80.000	0.114\\
};
\addlegendentry{\AdPMTTT};
\end{axis}
\end{tikzpicture}%
\caption{The performance of the d-PM as a function of $\snr$ for a fixed number of samples $N=100$.} 
\label{fig:d-pm-snr}
\end{figure}
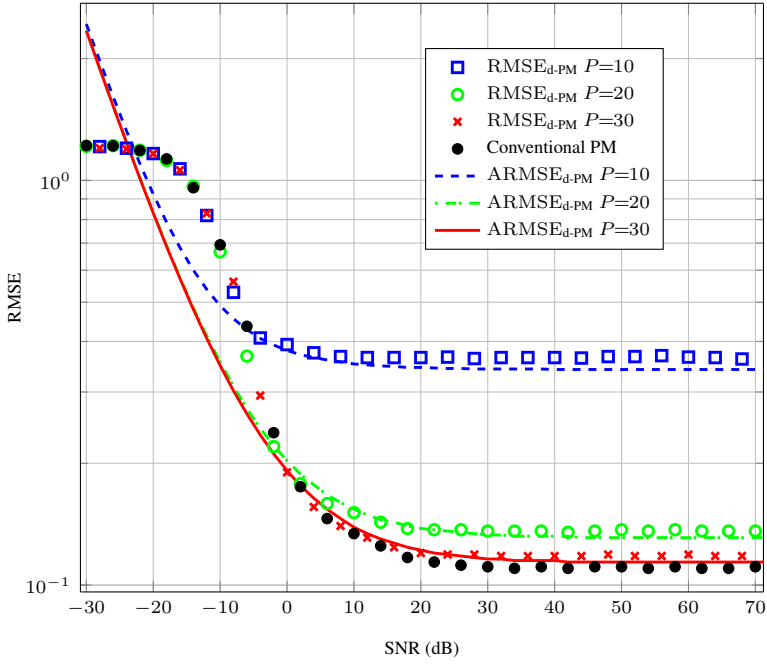
\begin{figure}[t]
\centering
\definecolor{mycolor1}{rgb}{0.00000,0.49804,0.00000}%
\begin{tikzpicture}

\begin{axis}[%
width=\mfw\columnwidth,
scale only axis,
separate axis lines,
every outer x axis line/.append style={black},
every x tick label/.append style={font=\scriptsize\color{black}},
xmin=0,
xmax=1000,
xlabel={\NumberofSamples},
xmajorgrids,
every outer y axis line/.append style={black},
every y tick label/.append style={font=\scriptsize\color{black}},
ymode=log,
ymin=0.01,
ymax=1,
yminorticks=true,
ylabel={\NormalizedRMSE},
ylabel style={yshift=-10pt},
ymajorgrids,
yminorgrids,
legend style={at={(0.01,0.01)},anchor=south west,legend cell align=left,align=left,fill=white}
]
\addplot [color=blue,line width=\mlw pt, mark repeat = 2, mark phase = {2}, only marks,mark=square,mark options={solid}]
  table[row sep=crcr]{%
30	0.479151951290832\\
50	0.420861327989245\\
100	0.372480523799753\\
200	0.353537163366288\\
300	0.340538847462321\\
400	0.338001323250243\\
500	0.339449571052992\\
600	0.339177495495889\\
700	0.335148608143769\\
800	0.335633423133184\\
900	0.336147555343961\\
1000	0.33022051737355\\
};
\addlegendentry{\dPMT};

\addplot [color=green, line width=\mlw pt, mark repeat = 2, mark phase = {1}, only marks,mark=o,mark options={solid}]
  table[row sep=crcr]{%
30	0.261822099621013\\
50	0.207162078975044\\
100	0.153915929330729\\
200	0.118646982557424\\
300	0.097272579614178\\
400	0.096283970909501\\
500	0.0891553748877466\\
600	0.087569602310438\\
700	0.0799156818186397\\
800	0.0778694838053284\\
900	0.0772423742797605\\
1000	0.0799168380541522\\
};
\addlegendentry{\dPMTT};

\addplot [color=red,line width=\mlw pt, only marks,mark repeat = 2, mark phase = {2}, mark=x,mark options={solid}]
  table[row sep=crcr]{%
30	0.253527342001843\\
50	0.196574230909917\\
100	0.139339206953051\\
200	0.0990071096565545\\
300	0.0812285180952988\\
400	0.0706817115431207\\
500	0.0635185300797248\\
600	0.0582557854755351\\
700	0.0541846316891399\\
800	0.0509180973596312\\
900	0.0482247309082797\\
1000	0.0459565261648101\\
};
\addlegendentry{\dPMTTT};

\addplot [color=black, line width=0.1 pt, mark repeat=2, mark phase = {1} ,only marks,mark=*,mark options={solid}]
  table[row sep=crcr]{%
10.000	0.394\\
20.000	0.278\\
50.000	0.167\\
100.000	0.121\\
200.000	0.087\\
300.000	0.069\\
400.000	0.062\\
500.000	0.057\\
600.000	0.048\\
700.000	0.047\\
800.000	0.044\\
900.000	0.040\\
1000.000	0.039\\
};
\addlegendentry{\dPMCent};

\addplot [color=blue,dashed, line width=\mlw pt, ]
  table[row sep=crcr]{%
30	0.411177197165535\\
50	0.378724177448723\\
100	0.352428700478113\\
200	0.338515852049695\\
300	0.333749376746001\\
400	0.331340427163742\\
500	0.329886613838566\\
600	0.328913835020295\\
700	0.328217227819738\\
800	0.327693800620505\\
900	0.327286111831165\\
1000	0.326959594787373\\
};
\addlegendentry{\AdPMT};

\addplot [color=green,dash pattern=on 1pt off 3pt on 4pt off 1pt,line width=\mlw pt]
  table[row sep=crcr]{%
30	0.238959463674703\\
50	0.200084292789944\\
100	0.146468902778636\\
200	0.104673969730738\\
300	0.092249516361824\\
400	0.086805869222887\\
500	0.082245781653573\\
600	0.0798153157581974\\
700	0.0751513879626043\\
800	0.0734968127737428\\
900	0.0721856238353247\\
1000	0.0705474420468166\\
};
\addlegendentry{\AdPMTT};

\addplot [color=red,solid,line width=\mlw pt]
  table[row sep=crcr]{%
30	0.253153534859252\\
50	0.196091884908684\\
100	0.142657901554583\\
200	0.0980459424543421\\
300	0.0800541767878069\\
400	0.0693289507772913\\
500	0.0620096986987041\\
600	0.056606851268965\\
700	0.0524077606896431\\
800	0.049022971227171\\
900	0.0462193005181942\\
1000	0.0438474784491883\\
};
\addlegendentry{\AdPMTTT};

\end{axis}
\end{tikzpicture}%
\caption{The performance of the d-PM as a function of $N$ for a fixed $\snr=10$ dB.} 
\label{fig:d-pm-n}
\end{figure}
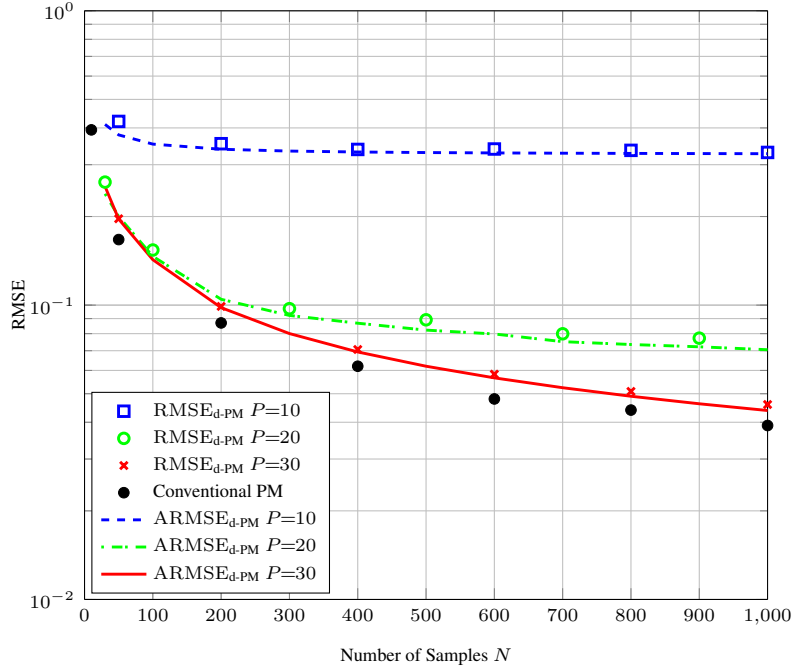
\begin{figure}[t]
\centering
\definecolor{mycolor1}{rgb}{0.00000,0.49804,0.00000}%
\definecolor{mycolor2}{rgb}{1.00000,0.00000,1.00000}%
\begin{tikzpicture}

\begin{axis}[%
width=\mfw \columnwidth,
scale only axis,
separate axis lines,
every outer x axis line/.append style={black},
every x tick label/.append style={font=\scriptsize\color{black}},
xmin=-20.4599020654475,
xmax=70.4585535151977,
xlabel={\SNR},
xmajorgrids,
every outer y axis line/.append style={black},
every y tick label/.append style={font=\scriptsize\color{black}},
ymode=log,
ymin=6.76703378240473e-06,
ymax=35.1767777925381,
yminorticks=true,
ylabel={\RMSE },
ylabel style={yshift=-10pt},
ymajorgrids,
yminorgrids,
legend style={at={(0.01,0.01)},anchor=south west,legend cell align=left,align=left,fill=white}
]
\addplot [color=blue,line width=\mlw pt, mark repeat = 2, mark phase = {2}, only marks,mark=square,mark options={solid}]
  table[row sep=crcr]{%
-30	28.910\\
-25	28.760\\
-20	27.329\\
-15	20.929\\
-10	11.461\\
-5	3.703\\
0	2.263\\
5	1.746\\
10	1.551\\
15	1.482\\
20	1.688\\
25	1.556\\
30	1.574\\
35	1.641\\
40	1.532\\
45	1.573\\
50	1.522\\
55	1.616\\
60	1.549\\
65	1.537\\
70	1.509\\
75	1.459\\
80	1.593\\
85	1.513\\
90	1.512\\
95	1.494\\
100	1.529\\
120	1.512\\
140	1.524\\
160	1.494\\
180	1.468\\
200	1.571\\
300	1.502\\
};
\addlegendentry{\dESPRITT};

\addplot [color=green, line width=\mlw pt, mark repeat = 2, mark phase = {1}, only marks,mark=o,mark options={solid}]
  table[row sep=crcr]{%
-30	28.994\\
-25	27.721\\
-20	26.210\\
-15	19.653\\
-10	11.540\\
-5	5.0010\\
0	1.7076\\
5	0.8534\\
10	0.6641\\
15	0.6096\\
20	0.5914\\
25	0.5800\\
30	0.5840\\
35	0.5918\\
40	0.5806\\
45	0.5855\\
50	0.5891\\
55	0.5912\\
60	0.5757\\
65	0.5903\\
70	0.5851\\
75	0.5719\\
80	0.5751\\
85	0.5788\\
90	0.5791\\
95	0.5732\\
100	0.5795\\
120	0.5888\\
140	0.5992\\
160	0.5770\\
180	0.5925\\
200	0.5734\\
300	0.5896\\
};
\addlegendentry{\dESPRITTT};

\addplot [color=red,line width=\mlw pt, only marks,mark repeat = 2, mark phase = {2}, mark=x,mark options={solid}]
  table[row sep=crcr]{%
-30	30.285\\
-25	27.884\\
-20	26.518\\
-15	21.110\\
-10	12.803\\
-5	3.8302\\
0	1.3023\\
5	0.6545\\
10	0.3629\\
15	0.2057\\
20	0.1749\\
25	0.1371\\
30	0.1373\\
35	0.1305\\
40	0.1292\\
45	0.128\\
50	0.1281\\
55	0.1289\\
60	0.1322\\
65	0.1289\\
70	0.1430\\
75	0.1409\\
80	0.1401\\
85	0.1423\\
90	0.1413\\
95	0.1399\\
100	0.1427\\
120	0.1435\\
140	0.1394\\
160	0.1421\\
180	0.1435\\
200	0.1395\\
300	0.1439\\
};
\addlegendentry{\dESPRITTTT};

\addplot [color=blue,dashed, line width=\mlw pt, ]
  table[row sep=crcr]{%
-30	773.098\\
-25	219.200\\
-20	70.3911\\
-15	50.3021\\
-10	15.6183\\
-5	6.246\\
0	3.145\\
5	2.253\\
10	1.758\\
15	1.491\\
20	1.392\\
25	1.394\\
30	1.398\\
35	1.391\\
40	1.398\\
45	1.397\\
50	1.390\\
55	1.396\\
60	1.396\\
65	1.360\\
70	1.374\\
75	1.343\\
80	1.266\\
85	1.334\\
90	1.305\\
95	1.351\\
100	1.259\\
120	1.296\\
140	1.262\\
160	1.347\\
180	1.252\\
200	1.329\\
300	1.333\\
};
\addlegendentry{\AdESPRITT};

\addplot [color=green,dash pattern=on 1pt off 3pt on 4pt off 1pt,line width=\mlw pt, smooth]
  table[row sep=crcr]{%
-30	771.2405\\
-25	216.0142\\
-20	70.2536\\
-15	30.2553\\
-10	10.30676\\
-5	4.1872\\
0	1.8076\\
5	1.021\\
10	0.720\\
15	0.585\\
20	0.527\\
25	0.515\\
30	0.519\\
35	0.527\\
40	0.516\\
45	0.521\\
50	0.524\\
55	0.526\\
60	0.511\\
65	0.526\\
70	0.520\\
75	0.507\\
80	0.510\\
85	0.514\\
90	0.514\\
95	0.508\\
100	0.515\\
120	0.524\\
140	0.534\\
160	0.512\\
180	0.528\\
200	0.509\\
300	0.525\\
};
\addlegendentry{\AdESPRITTT};

\addplot [color=red,solid,line width=\mlw pt, smooth]
  table[row sep=crcr]{%
-30	771.158\\
-25	215.872\\
-20	70.247\\
-15	29.253\\
-10	9.393\\
-5	3.576\\
0	1.373\\
5	0.635\\
10	0.328\\
15	0.201\\
20	0.140\\
25	0.122\\
30	0.113\\
35	0.112\\
40	0.108\\
45	0.108\\
50	0.107\\
55	0.110\\
60	0.107\\
65	0.109\\
70	0.108\\
75	0.106\\
80	0.105\\
85	0.108\\
90	0.106\\
95	0.105\\
100	0.108\\
120	0.109\\
140	0.105\\
160	0.107\\
180	0.109\\
200	0.105\\
300	0.109\\
};
\addlegendentry{\AdESPRITTTT};

\addplot [color=black, line width=0.1 pt, mark repeat=2, mark phase = {1} ,only marks,mark=*,mark options={solid}]
  table[row sep=crcr]{%
-30	33.7737\\
-25	30.124\\
-20	29.937\\
-15	22.277\\
-10	13.379\\
-5	2.960\\
0	1.248\\
5	0.671\\
10	0.343\\
15	0.198\\
20	0.111\\
25	0.063\\
30	0.036\\
35	0.020\\
40	0.011\\
45	0.006\\
50	0.003\\
55	0.002\\
60	0.001\\
65	0.000\\
70	0.000\\
75	0.000\\
80	0.000\\
85	6.166e-05\\
90	3.537e-05\\
95	2.075e-05\\
100	1.207e-05\\
120	1.182e-06\\
140	1.103e-07\\
160	1.133e-08\\
180	1.161e-09\\
200	1.126e-10\\
300	8.828e-15\\
};
\addlegendentry{\cESPRIT};

 \addplot [color=red!50!blue,mark repeat=2, mark phase={1},line width=\mlw pt,mark size=1 pt,only marks,mark=triangle,mark options={solid,rotate=180}]
  table[row sep=crcr]{%
-30	602.698\\
-25	177.977\\
-20	62.0684\\
-15	22.2759\\
-10	5.5318\\
-5	2.899\\
0	1.205\\
5	0.647\\
10	0.432\\
15	0.202\\
20	0.134\\
25	0.054\\
30	0.032\\
35	0.022\\
40	0.009\\
45	0.006\\
50	0.003\\
55	0.002\\
60	0.001\\
65	0.00066\\
70	0.00034\\
75	0.00020\\
80	9.95606e-05\\
85	6.7951e-05\\
90	3.51297e-05\\
95	1.8887e-05\\
100	1.1978e-05\\
120	1.0541e-06\\
140	1.0499e-07\\
160	0\\
180	7.96274847637202e-09\\
200	0\\
300	0\\
};
\addlegendentry{\Rao};

\addplot [color=mycolor2,mark repeat=2, mark phase={2}, line width=\mlw pt,mark size=1 pt,only marks,mark=triangle,mark options={solid}]
  table[row sep=crcr]{%
-30	487.854\\
-25	155.147\\
-20	49.926\\
-15	16.62\\
-10	6.015\\
-5	2.515\\
0	1.219\\
5	0.647\\
10	0.357\\
15	0.199\\
20	0.111\\
25	0.062\\
30	0.035\\
35	0.019\\
40	0.011\\
45	0.006\\
50	0.003\\
55	0.00198\\
60	0.00111\\
65	0.00062\\
70	0.00035\\
75	0.00019\\
80	0.000111\\
85	6.291e-05\\
90	3.5371e-05\\
95	1.9894e-05\\
100	1.1187e-05\\
120	1.1187e-06\\
140	1.1187e-07\\
160	1.118e-08\\
180	1.118e-09\\
200	1.118e-10\\
300	1.118e-15\\
};
\addlegendentry{\CRB};

\end{axis}
\end{tikzpicture}%
\caption{The performance of the d-ESPRIT algorithm as a function of $\snr$ for a fixed number of samples $N=100$.}
\label{fig:d-esprit-snr}
\end{figure}
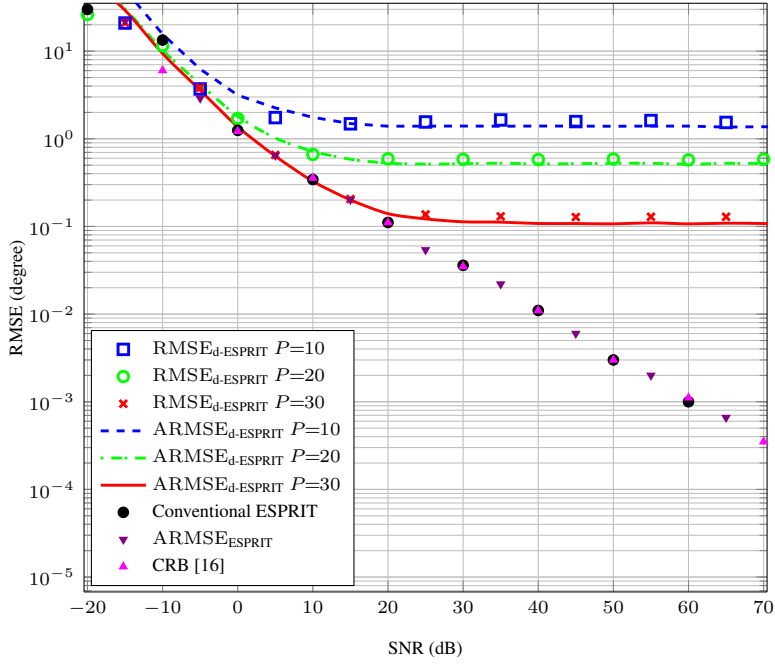
\begin{figure}[t]
\centering
\definecolor{mycolor1}{rgb}{0.00000,0.49804,0.00000}%
\definecolor{mycolor2}{rgb}{1.00000,0.00000,1.00000}%
\begin{tikzpicture}
\begin{axis}[%
width=\mfw\columnwidth,
scale only axis,
separate axis lines,
every outer x axis line/.append style={black},
every x tick label/.append style={font=\scriptsize\color{black}},
xmin=-2,
xmax=1008.43112194423,
xlabel={\NumberofSamples},
xmajorgrids,
every outer y axis line/.append style={black},
every y tick label/.append style={font=\scriptsize\color{black}},
ymode=log,
ymin=0.0966903343049588,
ymax=90.0341890727576,
yminorticks=true,
ylabel={\RMSE},
ylabel style={yshift=-10pt},
ymajorgrids,
yminorgrids,
legend style={at={(0.44,0.4)},anchor=south west,legend cell align=left,align=left,fill=white}
]
\addplot [color=blue,line width=\mlw pt, mark repeat = 2, mark phase = {1}, only marks,mark=square,mark options={solid}]
  table[row sep=crcr]{%
30	2.4999\\
50	2.2727\\
100	2.1501\\
200	2.1432\\
300	2.1839\\
400	2.2110\\
500	2.2107\\
600	2.2200\\
700	2.2758\\
800	2.1281\\
900	2.1097\\
1000	2.0732417\\
};
\addlegendentry{\dESPRITT}; 

\addplot [color=green, line width=\mlw pt, mark repeat = 2, mark phase = {1}, only marks,mark=o,mark options={solid}]
  table[row sep=crcr]{%
30	0.920\\
50	0.798\\
100	0.691\\
200	0.662\\
300	0.640\\
400	0.609\\
500	0.632\\
600	0.598\\
700	0.609\\
800	0.592\\
900	0.598\\
1000	0.597\\
};
\addlegendentry{\dESPRITTT};

\addplot [color=red,line width=\mlw pt, only marks,mark repeat = 2, mark phase = {2}, mark=x,mark options={solid}]
  table[row sep=crcr]{%
30	0.7316\\
50	0.5855\\
100	0.4112\\
200	0.3102\\
300	0.2576\\
400	0.2274\\
500	0.2079\\
600	0.1865\\
700	0.1705\\
800	0.1608\\
900	0.1562\\
1000	0.1534\\
};
\addlegendentry{\dESPRITTTT};

\addplot [color=blue,dashed, line width=\mlw pt, ]
  table[row sep=crcr]{%
30	 2.6927\\
50	 2.3253\\
100	 2.1035\\
200	 2.00472\\
300	 2.00383\\
400	 2.00339\\
500	 2.00312\\
600	 2.00295\\
700	 2.00282\\
800	 2.00272\\
900	 2.00265\\
1000 2.00259\\
};
\addlegendentry{\AdESPRITT};

\addplot [color=green,dash pattern=on 1pt off 3pt on 4pt off 1pt,line width=\mlw pt, smooth]
  table[row sep=crcr]{%
30	1.0789\\
50	0.9028\\
100	0.7418\\
200	0.6452\\
300	0.6094\\
400	0.5905\\
500	0.5789\\
600	0.5710\\
700	0.5653\\
800	0.5610\\
900	0.5576\\
1000	0.554\\
};
\addlegendentry{\AdESPRITTT};

\addplot [color=red,solid,line width=\mlw pt, smooth]
  table[row sep=crcr]{%
30	0.8707\\
50	0.6670\\
100	0.4734\\
200	0.3211\\
300	0.2593\\
400	0.2231\\
500	0.1988\\
600	0.1812\\
700	0.1677\\
800	0.1570\\
900	0.1483\\
1000	0.141\\
};
\addlegendentry{\AdESPRITTTT};

\addplot [color=black, line width=0.1 pt, mark repeat=2, mark phase = {1} ,only marks,mark=*,mark options={solid}]
  table[row sep=crcr]{%
30	0.6756\\
50	0.5192\\
100	0.3780\\
200	0.2611\\
300	0.2052\\
400	0.1866\\
500	0.150\\
600	0.1449\\
700	0.1349\\
800	0.1219\\
900	0.1174\\
1000	0.1146\\
};
\addlegendentry{\cESPRIT};

\addplot [color=red!50!blue,mark repeat=2, mark phase={2},line width=\mlw pt,mark size=1 pt,only marks,mark=triangle,mark options={solid,rotate=180}]
  table[row sep=crcr]{%
30	0.6310\\
50	0.4887\\
100	0.3456\\
200	0.2443\\
300	0.1995\\
400	0.1728\\
500	0.1545\\
600	0.1410\\
700	0.1306\\
800	0.1221\\
900	0.1152\\
1000	0.109\\
};
\addlegendentry{\Rao};

\addplot [color=mycolor2,mark repeat=2, only marks, mark phase={1}, line width=\mlw pt,mark size=1 pt,mark=triangle,mark options={solid}]
  table[row sep=crcr]{%
30	0.65198\\
50	0.50502\\
100	0.35710\\
200	0.25251\\
300	0.20017\\
400	0.16855\\
500	0.14870\\
600	0.13578\\
700	0.12897\\
800	0.12025\\
900	0.11603\\
1000 0.1129\\
};
\addlegendentry{\CRB};

\end{axis}
\end{tikzpicture}%
 
\caption{The performance of the d-ESPRIT algorithm as a function of $N$ for a fixed $\snr=10$ dB.}
\label{fig:d-esprit-n}
\end{figure}
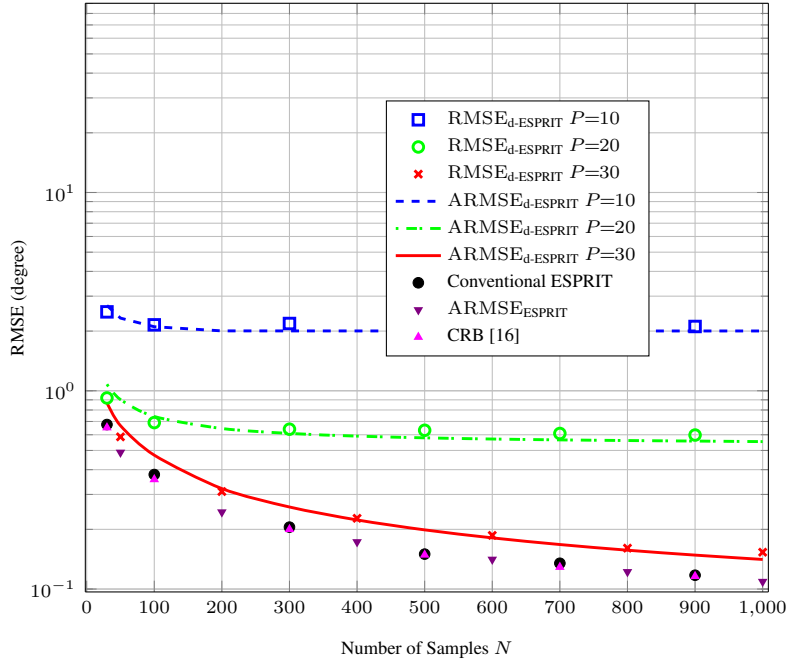

In Fig.~\ref{fig:d-esprit-n}, the $\snr$ is fixed to $10$ dB
and $\adespritrmse$ and $\despritrmse$ are computed for different number of samples $N$.
It is obvious that the error in the d-ESPRIT algorithm does not approach zero 
when $N \rightarrow \infty$, which is in accordance with our conclusion
in Sec.~\ref{sec:d-ESPRIT}, that the d-ESPRIT algorithm is not a consistent estimator of the DOAS,
unless the number of AC iterations $\itrac$ is infinitely large. 

In Fig.~\ref{fig:d-esprit-snr} and
Fig.~\ref{fig:d-esprit-n}, it can be observed 
that the values obtained for the averaged RMSE of d-ESPRIT algorithm $\despritrmse$
are similar to the results of the 
analytical expression $\adespritrmse$.

\section{Conclusions}
In this paper,  we derived an analytical expression 
for the second order statistics of the eigenvectors for 
the sample covariance matrix computed using the d-PM.
This analytical expression is used to derive the 
MSE of the DOA estimates obtained from
the d-ESPRIT algorithm.
It has been shown that 
the AC errors in d-PM and d-ESPRIT algorithm
are dominant when $N$ is very large or the $\snr$ is very high.
In our analysis, errors resulting from a small number of PM iterations is not considered.
Nevertheless, in the simulations,
it has been shown that $10$ PM iterations
are sufficient to make these errors negligibly small.

\appendix
\section{Proof of Theorem \ref{thm:delta-eigenvector}}
\label{apx:theorem2}

\begin{proof}
In order to prove Theorem~\ref{thm:delta-eigenvector},
the matrix $\pdsmcov$ is written in terms of $\smcov$ and $\avgmat{}$. 
Then, a first order analysis is carried out.
For convenience, we drop the dependency on $P$ from  $\pdsmcov$,
$\pdsmcoveigv{l}$ and $\pvavone{l}$ throughout the proof.

The largest eigenvalue of the matrix $\avgmat{}$ is $\avgmateige_1=1$
and its corresponding eigenvector 
is 
$\avgmateigv_1 = \frac{1}{\sqrt{K}} \ones{K}$, see Sec.~\ref{sec:ac}, thus
\begin{equation}
\begin{aligned}
\dsmcov &= K \left( \sensorsel \avgmat{}^{\itrac}  \sensorsel^T \right) \odot \smcov \\
       &=
       K \left( \sensorsel \sum_{k=1}^K  \avgmateige_k^{\itrac} \avgmateigv_k \avgmateigv_k^H
	   \sensorsel^T \right)  \odot \smcov \\
	   &= \smcov + 
	   K \left( \sum_{k=2}^K  \avgmateige_k^{\itrac} \sensorsel \avgmateigv_k \avgmateigv_k^H
	   \sensorsel^T \right)  \odot \smcov \\
	   &= \smcov + K \sum_{k=2}^{K} \alpha_k^{\itrac} 
	      \vamthree_k \,  \smcov \, \vamthree_k^H,
\end{aligned}
\label{eq:va-1}
\end{equation}
where 
$\vamthree_k=\diag{\sensorsel \pmb{\beta}_k }$
and, for the last equality, the rank one Hadamard product property \cite[p.~104]{johnson1990matrix} is used.
Note that the second term in Eq.~(\ref{eq:va-1}) accounts for the errors 
resulting from the finite number of AC iterations $\itrac < \infty$,
and that $\dsmcov\rightarrow\smcov$ when $\itrac\rightarrow\infty$.
Substituting Eq.~(\ref{eq:error-mcov-cent}) in Eq.~(\ref{eq:va-1})
yields
\begin{equation}
\begin{aligned}
\dsmcov & = 
       \mcov + \error \mcov + K \sum_{k=2}^{K} \avgmateige_k^{\itrac} 
	      \vamthree_k  \left(\mcov + \error \mcov \right) \vamthree^H_k \\
	& \approx \mcov + \error \mcov + K \sum_{k=2}^{K} \avgmateige_k^{\itrac} 
	      \vamthree_k \mcov \vamthree^H_k
\end{aligned}
\label{eq:va-3}
\end{equation} 
where the term
$\sum_{k=2}^{K} \avgmateige_k^{\itrac} \, \vamthree_k \, \error \mcov \, \vamthree^H_k$
is neglected in the approximation 
since the magnitudes of $\avgmateige_2  \mdots \avgmateige_K$ 
are all strictly smaller than one (see Sec.~\ref{sec:ac})
and they are multiplied with the small variation $\error \mcov$.
 
Multiplying Eq.~(\ref{eq:va-3}) from the right with $\dsmcoveigv_l$ and keeping the first order terms,
we find
\begin{equation}
\begin{aligned}
\dsmcov \, \dsmcoveigv_l  & \approx \left(
          \mcov + \error \mcov + K \sum_{k=2}^{K} \avgmateige_k^{\itrac} 
	        \vamthree_k \mcov \vamthree^H_k
	        \right) \left( \mcoveigv_l + \error \dsmcoveigv \right) \\
	        & \approx 
          \mcov \mcoveigv_l + \error \mcov \mcoveigv_l + \vavone_l + 
	        \mcov \error \dsmcoveigv,
\end{aligned}
\label{eq:va-4}
\end{equation}   
where $\vavone_l$ is defined in Eq.~(\ref{eq:def-vavone}).
The left hand side of Eq.~(\ref{eq:va-4}) can be written as 
$\dsmcov \, \dsmcoveigv_l = \dsmcoveige_l \, \dsmcoveigv_l$,
where $\dsmcoveige_l$ is the $l$th eigenvalue of $\dsmcov$.
Expressing $\dsmcoveige_l$ as a perturbation in $\mcoveige_l$, i.e., 
$\dsmcoveige_l=\mcoveige_l + \error\dsmcoveige_l$,
the left hand side of Eq.~(\ref{eq:va-4}) becomes
\begin{equation}
\begin{aligned}
\dsmcov \, \dsmcoveigv_l  & = \left( \mcoveige_l + \error\dsmcoveige_l \right) 
\left( \mcoveigv_l + \error \dsmcoveigv_l \right) \\
   & \approx \mcoveige_l \mcoveigv_l + \error\dsmcoveige_l  \mcoveigv_l + \mcoveige_l \error \dsmcoveigv_l,
\end{aligned}
\label{eq:va-5}
\end{equation}
where only first order terms are kept.    
Substituting Eq.~(\ref{eq:va-5}) in Eq.~(\ref{eq:va-4}) yields
\begin{equation}
\left( \mcov - \mcoveige_l \id_M \right) \error \dsmcoveigv_l
\approx 
          \error\dsmcoveige_l  \mcoveigv_l - \error \mcov \mcoveigv_l - \vavone_l
\label{eq:va-6}
\end{equation}
The matrix $\mcov - \mcoveige_l \id_M$ can be written as 
\begin{equation}
\begin{aligned}
\mcov - \mcoveige_l \id_M  & =  \sum_{k=1, k\neq l}^{M} \left( \mcoveige_k - \mcoveige_l \right) \mcoveigv_l \mcoveigv^H_k  \\ 
   &= \vamfour_{-l} \vamtwo_{-l} \,  \vamfour_{-l}^H
\end{aligned}
\label{eq:va-7}
\end{equation}
where 
$\vamfour_{-l}$ and 
$\vamtwo_{-l}$ are defined in Eq.~(\ref{eq:def-varmone}).
Thus multiplying Eq.~(\ref{eq:va-7}) by $\vamone_l \eq \vamfour_{-l} \vamtwo^{-1}_{-l} \,  \vamfour_{-l}^H$
we find
\begin{equation} 
\begin{aligned}
\error \dsmcoveigv_l
 & \approx 
         \vamone_l \left( \error\dsmcoveige_l  \mcoveigv_l - \error \mcov \mcoveigv_l - \vavone_l \right) \\
 & = -\vamone_l \left( \error \mcov \mcoveigv_l + \vavone_l \right).
\end{aligned}
\label{eq:va-8}
\end{equation}
\end{proof}

\section{Proof of Theorem \ref{thm:mse}}
\label{apx:theorem3}

\begin{proof}
In order to prove Theorem~\ref{thm:mse}, we compute 
$\sosone$
and $\sostwo$ 
based on the expression of $\error \pdsmcoveigv{l}$ 
which we found in Theorem~\ref{thm:delta-eigenvector}.
Then, results from \cite{brillinger2001time} are used to simplify the computed expression.
For convenience, we drop the dependency on $P$ from 
$\pdsmcoveigv{l}$ and $\pvavone{l}$ throughout the proof.
  
Using the result from Theorem \ref{thm:delta-eigenvector}, 
$\sosone$
and $\sostwo$  are written as
\begin{equation}
\begin{aligned}
\sosone
 & \approx 
         \vamone_l  
         \expect { 
         (\error \mcov \mcoveigv_l + \vavone_l)
         (\mcoveigv_m^H \error \mcov  + \vavone_m^H)
         }
         \vamone_m^H \\
 & =    \expect { 
         \vamone_l \error \mcov \mcoveigv_l \mcoveigv_m^H \error \mcov \vamone_m^H 
         }  +  \vamone_l \vavone_l \vavone_m^H \vamone_m^H    
\end{aligned}
\label{eq:va-3-1}
\end{equation}
and
\begin{equation}
\begin{aligned}
\sostwo
 & \approx 
         \vamone_l  
         \expect { 
         (\error \mcov \mcoveigv_l + \vavone_l)
         (\mcoveigv_m^T \error \mcov  + \vavone_m^T)
         }
         \vamone_m^T \\
 & =    
\expect { 
         \vamone_l \error \mcov \mcoveigv_l \mcoveigv_m^T \error \mcov \vamone_m^T 
         }  +  \vamone_l \vavone_l \vavone_m^T \vamone_m^T.
\end{aligned}
\label{eq:va-3-2}
\end{equation}
Using the following results from Theorem 9.2.2 in \cite{brillinger2001time}\footnote{See also the proof of the Theorem 9.2.2 \cite[p.~454]{brillinger2001time}.} 
$$
\expect {
         \vamone_l \error \mcov \mcoveigv_l \mcoveigv_m^H \error \mcov \vamone_m^H 
         }
= \dfrac{\mcoveige_l}{N}
\sum_{\underset{\scriptstyle{i\neq l}}{i=1}}^{M}
\dfrac{\mcoveige_i}{\left( \mcoveige_l - \mcoveige_i \right)^2} 
 \mcoveigv_l \mcoveigv_l^H 
 \chrod{l,m}
$$
and
$$
\expect { 
         \vamone_l \error \mcov \mcoveigv_l \mcoveigv_m^T \error \mcov \vamone_m^T 
         } =
\dfrac{\mcoveige_l \mcoveige_m}{N}
\dfrac{\mcoveigv_l \mcoveigv_m^T}{\left( \mcoveige_l -\mcoveige_m \right)^2}  (\chrod{l,m} - 1) 
$$
in Eq.~(\ref{eq:va-3-1}) and Eq.~(\ref{eq:va-3-2}) proves the theorem.
 
 \end{proof}

\bibliographystyle{plain}

\end{document}